\documentclass[12pt]{amsart}
\usepackage{amssymb}
\usepackage{latexsym}
\usepackage{amsfonts,euscript}
\usepackage{tabularx}
\usepackage{array}
\usepackage{booktabs}
\usepackage{tikz}
\usepackage{ctable}
\usetikzlibrary{shapes,arrows}
\usetikzlibrary{positioning,calc}
\usepackage{verbatim}
\usepackage{caption}
\usepackage{subcaption}

\allowdisplaybreaks[1]

\usepackage{lineno}
\usepackage{graphicx}
\usepackage{sidecap}
\usepackage{epstopdf}
\usepackage{epsf}
\usepackage{epsfig}

\usepackage{graphicx}
\makeatletter \newcommand*{\rom}[1]{\expandafter\@slowromancap\romannumeral #1@}

\newtheorem{theorem}{Theorem}[]

\newtheorem{proposition}{Proposition}[section]

\theoremstyle{definition}

\numberwithin{equation}{section}

\newcommand{\D}{\displaystyle}

\newcommand{\bfa}[1]{\mbox{\boldmath $ #1 $}}
\newcommand{\vo}{\vec{o}\@ifnextchar{^}{\,}{}}

\begin{document}
\bibliographystyle{abbrv}
\title{An Immuno-Epidemiological Vector-Host Model with Within-Vector Viral Kinetics } 

\author{Hayriye Gulbudak$^*$}
\address{Mathematics Department,
University of Louisiana at Lafayette,
Lafayette, LA }
\email{hayriye.gulbudak@louisiana.edu}

\thanks{$^*$author for correspondence}

\begin{abstract}
A current challenge for disease modeling and public health is understanding pathogen dynamics across scales since their ecology and evolution ultimately operate on several coupled scales. This is particularly true for vector-borne diseases, where within-vector, within-host, and between vector-host populations all play crucial roles in diversity and distribution of the pathogen. Despite recent modeling efforts to determine the effect of within-host virus-immune response dynamics on between-host transmission, the role of within-vector viral dynamics on disease spread is overlooked. Here we formulate an age-since-infection structured epidemic model coupled to nonlinear ordinary differential equations describing within-host immune-virus dynamics and within-vector viral kinetics, with feedbacks across these scales. We first define the \emph{within-host viral-immune response and within-vector viral kinetics dependent} basic reproduction number $\mathcal R_0.$ Then we prove that whenever $\mathcal R_0<1,$ the disease free equilibrium is locally asymptotically stable, and under certain biologically interpretable conditions, globally asymptotically stable. Otherwise if $\mathcal R_0>1,$ it is unstable and the system has a unique positive endemic equilibrium.  In the special case of constant vector to host inoculum size, we show the positive equilibrium is locally asymptotically stable and the disease is weakly uniformly persistent.  Furthermore numerical results suggest that within-vector-viral kinetics and dynamic inoculum size may play a substantial role in epidemics. Finally, we address how the model can be utilized to better predict the success of control strategies such as vaccination and drug treatment.
\bigskip

\end{abstract}
\date{\today}
 \maketitle
\pagestyle{myheadings}
\markboth{\sc Impact of Within-Mosquito Viral Dynamics on the Spread of Arbovirus Diseases}
{\sc }

\baselineskip14pt


\section{Introduction}

The ecology and evolution of infectious diseases operate on several interdependent scales. This is particularly true for vector-borne diseases, where coupled within-vector, within-host, and between vector-host population dynamics together determine the diversity and distribution of the pathogen. One of the major mechanisms in determining disease abundance and virus evolution is within-vector viral kinetics \cite{forrester2014arboviral,sim2014mosquito,tabachnick2013nature}. Yet, from the mathematical point of view, the within-vector viral dynamics has been overlooked, and rarely studied \cite{reiner2013systematic,rock2015age,wang2017global}. There is also a serious need for an integrated modeling approach that links all scales \cite{handel2015crossing}. However, traditional modeling approaches treat within-host, within-vector, and between-host pathogen dynamics as separate systems.  Multi-scale mathematical models can be both an avenue for understanding the complex features displayed by vector-borne diseases and tools for targeting interventions.

More than $17\%$ of all infectious diseases are accounted to be vector-borne diseases, causing more than $700, 000$ deaths annually world-wide according to World Health Organization (WHO). In particular, a vast majority of vector-borne vertebrate infecting viruses (arboviruses) are responsible for a number of severe diseases in humans (yellow fever (YFV), dengue (DENV), various encephalitides, etc.) and livestock (West Nile encephalomyelitis (WNV), Rift Valley fever (RVF), vesicular stomatitis, etc.). 
Mathematical modeling can help us understanding the impact of mechanisms behind the establishment and transmission of vector-borne viruses, which are crucial for developing effective intervention strategies. 

Mosquito-borne diseases are spread when mosquitoes bite hosts and release microscopic parasites, which live in the salivary glands of the mosquitoes, into the hosts' bloodstream. In a recent work, Churcher et al. \cite{churcher2017probability} show that the number of parasites each mosquito carries influences the chance of successful malaria infection. In particular,  it demontrates that the more parasites present in a mosquito's salivary glands, the more likely it was to be infectious, and also the faster any infection would develop, highlighting the importance of within-mosquito viral kinetics. 
So far disease control authorities, including WHO (World Health Organization), has relied on the average number of potentially infectious mosquito bites per person per year. However, not every infectious mosquito bite will result in disease, and not all equally infectious. Therefore, there is an urgent need for unified models to understand how the within-vector viral kinetics can be scaled up to the host population level disease transmission.

A large number of previous studies formulate and analyze coupled nonlinear ODE models, describing population level vector-host disease spread. In a study that extensively reviewed mosquito-borne pathogen models, Reiner et al.\cite{reiner2013systematic} suggests that  ``moving forward in mosquito-borne disease modeling and addressing public health challenges will require modeling efforts on the heterogeneities such as variation in individual hosts and mosquitoes and their consequences for heterogeneous biting."
Some of the previous vector-borne disease modeling studies focus on the impact of mechanisms such as temperature and rainfall on the vector population \cite{eikenberry2018mathematical, taghikhani2018mathematics}, or preventative measures, such as mosquito reduction strategies or personal protection \cite{prosper2014optimal}.  
Due to the extrinsic and intrinsic incubation periods and vector maturation process, time delayed vector-borne disease models have been deployed in \cite{cai2017global,martcheva2013unstable,nah2014malaria,fan2010impact}. Dynamical properties of vector-borne diseases are also studied in age-since-\textit{host} infection structured PDE models with direct transmission, time delay or reinfection \cite{wei2008epidemic, lashari2011global,cai2013dynamical, xu2016hopf}.  Furthermore, in within-host scale,  arbovirus-immune response dynamics such as within-host DENV transmission is modeled and analyzed in \cite{ciupe2017host}.  
Immunology is also coupled with between vector-host disease dynamic models in prior works \cite{gulbudak2017vector, tuncer2016structural}. Gulbudak et al. \cite{gulbudak2017vector} construct a multi-scale ODE-PDE hybrid model (similar to \cite{gilchrist2002modeling}), coupling within-host viral-immune response and between vector-host disease transmission on population scale and study coevolution of both vector-borne pathogen and host.  In a companion paper, Tuncer et. al.\cite{tuncer2016structural} fit RVF multi-scale data  to the multi-scale model in \cite{gulbudak2017vector}, and utilize identifiability analysis for the model parameters. Furthermore, in a recent work, the effect of antibody-dependent enhancement on the transmission dynamics and persistence of multiple strain pathogens is also investigated in a two strain immuno-epidemiological DENV model structured by dynamic antibody size \cite{DENVimmunoepi2019}.

Within-vector viral kinetics are more often overlooked in disease dynamic models, despite its role on disease transmission and virus evolution \cite{tabachnick2013nature,reiner2013systematic}. Many prior modeling studies do not take into account how infectious each of those bites and each bite is considered equally infectious. In a recent article, Rock et al.\cite{rock2015age} develop and numerically study an age-since-vector infection and -bite model for vector-borne diseases and illustrate the distinct dynamics induced by complex interaction between feeding and life-expectancy, vector-infection-age, and bite distributions; yet the heterogeneity among host infectivity is ignored. Similarly, in another study, Wang et al. \cite{wang2017global} formulate and analyze a hybrid hyperbolic ODE-PDE model, where vector and host are stratified by infection ages. Nevertheless, to our best knowledge,  none of the modeling studies exclusively coupled between vector-host disease dynamics, within-vector viral dynamics and within-host virus-immune response, that are crucial competences in disease spread.

Our modeling framework, introduced here, allows for variable vector to host inoculum size and infectivity of mosquitoes, both dependent on within-vector viral kinetics, and both inducing heterogeneity among immune-pathogen dynamics and the infected host population. The multi-scale framework also enables direct incorporation of within-vector viral data, which can be utilized to assess the role of within-vector viral dynamics on the disease spread. As opposed to some other multi-scale models \cite{agusto2019transmission,garira2019coupled,kitagawa2019mathematical}, our PDE-ODE system does not reduce to a simpler set of ODEs because both the within-host and within-vector components operate on a different timescale with dynamics stratifying the structured PDE epidemic model.   Interestingly, analytical and numerical results show that the within-vector-viral dynamics can be a \emph{driving mechanism} in large epidemics. The model here can be applied to many arbovirus diseases including WNV, DENV, RVFV, and may help us understanding the role of mosquitoes, environmental factors, and host immune response on disease dynamics and the impact of disease control strategies, such as vaccination, drug treatment, and Wolbachia biocontrol strategy. It scales within-vector viral kinetics and within-host viral-immune response process up to the host population level disease dynamics. 

Here, we organize the paper as follows: in section \eqref{acc_scales}, we present an immuno-epidemiological model, incorporating within-vector viral kinetics, described with a modified logistic model with Allee effect.  Distinct from the existent studies, the modeling formulation that we introduced here not only tracks the \emph{infection-immune states of host population}, but also the \emph{heterogeneity among infectivity of vectors}, depending on \emph{their within-viral kinetics}. In section \eqref{analytical}, we define the within-host immune-response and within-vector-viral kinetics dependent reproduction number $\mathcal R_0$ and investigate the threshold properties of the multi-scale system such as local and global stability of the equilibria in addition to the disease persistence for a special case.  In section \eqref{implications}, we introduce a feasible way to incorporate ``epidemiological feedback'' between epidemiological and immunological scales. By doing so, we numerically assess the impact of both mechanisms \emph{within-vector viral kinetics} and \emph{host infectivity} on the population scale disease transmission. Finally, we summarize the results in Conclusion section \eqref{conc}.


\begin{table}[h]
\caption{Definition of the within-host model variables and parameters}
\label{table:immuniological parameters} 
\centering 
\begin{tabularx}{\textwidth}{>{} lX}
\toprule
Variable/Parameter  & Meaning \\ [0.5ex]
\toprule
\\
$P_0^{v\rightarrow h}=\hat{c}V(s,V_0)$ & Vector to host inoculum size (\emph{the amount of pathogen that an infectious mosquito, that has infection age $s$, injects to host upon giving a bite}) \\[0.5ex]
$P(\tau,s)$ & Pathogen concentration at $\tau$ days post host infection \emph{with initial condition $(P_0=P_0^{v\rightarrow h}, M_0, G_0),$}  \\[0.5ex]
$M(\tau,s)$ & The concentration of IgM immune response antibodies at host infection age $\tau,$ \emph{with initial condition $(P_0^{v\rightarrow h}, M_0, G_0),$} \\[0.5ex]  
$G(\tau,s)$ &  The concentration of IgG immune memory antibodies at host infection age $\tau,$ \emph{with initial condition $(P_0^{v\rightarrow h}, M_0, G_0),$} \\[0.5ex]
$r_h$ & Within-host parasite growth rate, \\[0.5ex] 
$K_h$ & Within-host parasite carrying cappacity, \\[0.5ex] 
$\epsilon$ & The efficiency of the IgM immune response at killing the parasite,\\[0.5ex] 
$\delta$ & The efficiency of the IgG immune response at killing the parasite, \\[0.5ex]
$a$ & The IgM immune response activation rate, \\[0.5ex]
$q$ & The per-capita rate at which the IgM immune response antibody production switches to the IgG immune response antibody production, \\[0.5ex]  
$b$ & The IgG activation rate upon coming into contact with the pathogen, \\[0.5ex]
$c$ & IgM immune response antibody decay rate, \\[0.5ex]
\bottomrule
\end{tabularx}
\end{table}

\begin{table}[h]
\caption{Definition of the within-vector model variables and parameters}
\label{table:immunological parameters} 
\centering 
\begin{tabularx}{\textwidth}{>{} lX}
\toprule
Variable/Parameter  & Meaning \\ [0.5ex]
\toprule
\\
$V_0^{h\rightarrow v}$ & Host to vector inoculum size (\emph{the amount of pathogen acquired from the blood meal of an infectious host by a susceptible vector upon giving a bite}) \\[0.5ex]
$V(s)=V(s,V_0^{h\rightarrow v})$ & Pathogen concentration within-an infectious vector  at $s$ days post infection \\[0.5ex]
$r_v$ & Within-vector pathogen growth rate, \\[0.5ex] 
$U_v$ & Allee threshold,\\[0.5ex] 
$K_v$ & Within-vector pathogen carrying cappacity. \\[0.5ex]
\bottomrule
\end{tabularx}
\end{table}

\section{Coupling disease dynamics across the scales}\label{acc_scales}

\noindent \emph{\textbf{\underline{Within-vector:}}} Consider a modified logistic model with \emph{Allee effect}, describing the within-vector viral dynamics:

\vspace{-.2cm}
{\small \begin{align}
\dot{V}(s) = r_v V(1-V/ K_v)(V / U_v-1), \label{within_vector} 
\end{align}}\vspace{-.2cm}

\noindent where the variable $V(s)$ represents the pathogen concentration within an infected mosquito at infection age $s;$ i.e. $s$ days passed after vector infection. The parameters $r_v, K_v$ represent the intrinsic pathogen growth rate within a vector, and carrying capacity, respectively.  If the initial  pathogen density, denoted by $V_0^{h\rightarrow v}(=V(0))$ (\emph{the amount of pathogen acquired in the blood meal of an infectious host by a susceptible vector upon a bite}), is less than $U_v,$ then the virus within a vector (mosquito) eventually clears. Otherwise if $V_0^{h\rightarrow v}>U_v,$ then the viruses persist and asymptotically converge to the carrying capacity $K_v.$ The model parameters are fitted to within-mosquito WNV viral data given in \cite{fortuna2015experimental} in Fig.2.  

\noindent  The Allee effect in the model refers that when host to vector inoculum size $V_0^{h \rightarrow v}$ is below $U_v$ (which is assume to be a small threshold size), the virus clears faster due to loss of viruses during transportation of them from foregut to mosquito midgut and mosquito immune response \cite{sim2014mosquito}. However when the inoculum size $V_0^{h \rightarrow v}>U_v,$ despite the loss during transportation, due to larger initial density, it replicates faster and reaches to carrying capacity $K_v$ \cite{franz2015tissue}.

\noindent Here, in particular,  we focus on mosquito-borne viruses such as WNV (Cx. Pipiens), DENV ( Aedes aegypti), and RVFV (most commonly transmited  by the Aedes and Culex mosquitoes). Arbovirus infection of mosquitoes is generally asymptomatic and persists for the life of the vector, but the virus appears to be continually targeted by innate immune response \cite{sim2014mosquito}. Here for simplicity, we do not consider mosquito immune response.

\noindent \emph{\textbf{\underline{Within-host viral-immune response dynamics:}}} To capture the short-term and long-term host-immune response to pathogen ($P$) introduction, we model two particular antibodies IgM ($M$) and IgG ($G$), released by B-cell lymphocytes.  The within-host infection model takes the following form (\cite{gulbudak2017vector, tuncer2016structural}):
{\small \begin{align} \label{PB}
P_{\tau}&=(\tilde{f}(P)-\epsilon M-\delta G)P, \quad M_{\tau}=\left(aP-(q+c)\right)M, 
\quad G_{\tau} =qM + bGP,  \\ \quad  P(0,s)&=h(V(s, V_0)), M(0,s)=M_0, G(0,s)=G_0.\nonumber 
\end{align}}
\vspace{-.55cm}

\noindent
 Note that $s$ in \eqref{PB} in the initial condition $(P(0,s), M(0,s), G(0,s)),$ means that the solutions to system \eqref{PB} depend on the vector to host inoculum size, $P_0^{v\rightarrow h}.$ In particular, $P_0 =P_0^{v\rightarrow h}=h(V(s, V_0)).$  The solutions are denoted by $P(\tau,s)$, $M(\tau,s)$ and $G(\tau,s),$ where the time variable $\tau$ refers time-since-infection within a host, and $s$ denotes vector infection age.
Note that upon \emph{viral progression within an infected mosquito midgut}, the amount of pathogen in an \emph{infected mosquito saliva} dynamically changes \cite{salas2015viral}, determining the vector to host inoculum size $P_0^{v\rightarrow h}$ which is  \emph{the amount of pathogen that is injected to a susceptible host by an infectious vector.}  
If the mosquito has $V(s, V_0)$ amount of pathogen within at the time giving a bite to a host, then we assume that the vector to host inoculum size, $P_0^{v\rightarrow h},$  is a function of the amount of the pathogen within-vector, $V(s, V_0),$ i.e. $P_0^{v\rightarrow h}=h(V(s, V_0))$.  Here we consider $h(V(s, V_0))=\hat{c} V(s, V_0),$ with a constant $\hat{c} \in (0,1)$.   The intuitive assumption is that the more pathogen an infectious mosquito has within-its midgut, the more pathogen it can inject to its hosts.  Furthermore, Fortuna et al. \cite{fortuna2015experimental} displays multiple experimental data suggesting that there is a linear relationship between the amount of pathogen within a vector and the amount of pathogen in the saliva of this mosquito.

\noindent We assume that the parasite replicates with a logistic growth rate $\tilde{f}(P)=r_h(1-P/ K_h),$ where the parameters $r_h$ and $K_h$ represent net viral growth rate, and the carrying capacity within a host, respectively.  Upon exposure to virus, the IgM immune response activates at a rate $a,$ and decays at a rate $c.$  The IgM immune response antibodies are responsible for \emph{rapid destruction of virus} and kills the pathogen at a rate $\epsilon.$ Furthermore, B-cells switch production of IgM antibodies to IgG antibodies with a per-capita rate $q$ \cite{honjo2002molecular}. The pathogen also stimulates the IgG antibody immune response that activates at a rate $b$ and kills the pathogen at a rate $\delta.$ Yet, the IgG immune response antibodies mainly responsible for \emph{life-long immunity}.

All parameters and dependent variables of this within-host model and their definitions are given in Table \ref{table:immunological parameters}.

\begin{theorem}\label{withhost} If $M_0>0 (\text{ or } G_0>0$), then the pathogen (within-the host) eventually clears ($\lim_{\tau \rightarrow \infty}P(\tau,s)=0$), the IgM immune response antibodies decays to zero after viral clearance, and subsequently the IgG immune memory antibodies reach a steady-state; i.e.
$\ \lim_{\tau \rightarrow \infty}M(\tau,s)=0 \ \text{and} \ \lim_{\tau \rightarrow \infty}G(\tau,s)=G^+,$ where $G^+>0$ depends on the initial condition; i.e. $G^+=z(P_0^{v\rightarrow h},M_0,G_0).$
\end{theorem}

\begin{proof} Let $P_0^{v \rightarrow h}(=P(0,s))>0.$ By the first equation in (\ref{PB}), we obtain 
 \begin{align}
		\label{Plim}
		P(\tau,s)\leq P_0^{v \rightarrow h}e^{\int_0^\tau{\left[r_h-\delta G(\sigma,s)\right]}d\sigma},
		\end{align}
where $s$ is fixed.  Without loss of generality, assume that  $M_0>0.$ Then $\exists \tilde{\tau}:$ $\frac{\partial}{\partial\tau}G(\tau,s) \geq 0$ for all $\tau \geq \tilde{\tau}.$ Therefore $G(\tau,s)$ increases for all $\tau \geq \tilde{\tau}$.\\
 
\noindent Case (i) Assume $G(\tilde{\tau},s)>\dfrac{r_h}{\delta}$. Then, we obtain $G(\tau,s)>\dfrac{r_h}{\delta},$ for all $\tau > \tilde{\tau}.$ Therefore as $\tau \rightarrow \infty$, the RHS of the inequality (\ref{Plim}) goes to zero. Then by comparison principle, we obtain $\lim_{\tau \rightarrow \infty}P(\tau,s)=0.$ Thus $\lim_{\tau \rightarrow \infty}M(\tau,s)=0$. Then $G(\tau,s)$ saturates as $\tau \rightarrow \infty;$ i.e $\lim_{\tau \rightarrow \infty}G(\tau,s)=\bar{G},$ for some $\bar{G}>0,$ depending on the initial condition $(P_0^{v \rightarrow h}, M_0, G_0).$ \\

\noindent Case (ii)Now suppose that $G(\tilde{\tau},s)<\dfrac{r_h}{\delta}.$ Assume that there exists $\hat{\tau}$: $P(\hat{\tau},s)=0. $ Then we obtain that $P(\tau,s)=0,$ for all $\tau> \hat{\tau}$ and $\lim_{\tau \rightarrow \infty}M(\tau,s)=0$ and $\lim_{\tau \rightarrow \infty}G(\tau,s)=\bar{G}$, for some $\bar{G}>0.$ Now assume that $P(\tau,s)>0,$ for all $\tau>0.$ Then $\frac{\partial}{\partial\tau}G(\tau,s)>0,$ for all $\tau>0.$ Hence there exists $\tau^{+}>0: \ G(\tau^{+})>\dfrac{r}{\delta}$. The rest of the proof follows the argument in case (i), completing the proof.
\end{proof}

\begin{table}[h]
\caption{Definition of the between-host model variables  and parameters}
\centering 
\begin{tabularx}{\textwidth}{>{} lX}
\toprule
Variable/Parameter  & Meaning \\ [0.5ex]
\toprule
\\
$S_v(t)$ & The number of susceptible vectors at time $t$, \\[0.5ex]
$i_v(t,s)$ &  The density  of infected vectors with infection age $s$ at time $t$, \\[0.5ex]
$S_H(t)$ & The number of susceptible hosts at time $t,$ \\[0.5ex]
$ i_H(t,\tau,s)$ & The density of the infected hosts (whom infected with a vector with infection age $s$) with host infection age $\tau$ at time $t,$\\[0.5ex]
$ R_H(t)$ & The number of recovered hosts at time $t,$ \\[0.5ex]
$ \Lambda$ & Susceptible host recruitment rate \\[0.5ex]
$ \eta$ & Susceptible vector recruitment rate  \\[0.5ex]
$\beta_v(s)$ & Infected vector transmission rate at $s$ days post infection \\[0.5 ex]
$\beta_H(\tau,s)$  & Infected host transmission rate (whom infected with a vector with infection age $s$) at $\tau$ days post infection \\[0.5 ex] 
$\nu_H(\tau,s)$ & Additional host mortality rate (whom infected with a vector with infection age $s$) due to disease at $\tau$ days post infection  \\[0.5 ex]	
$\gamma_H (\tau,s)$ & Per capita host recovery rate (whom infected with a vector with infection age $s$) at $\tau$ days post infection \\[0.5 ex]
$d$ & Host natural death rate  \\[0.5 ex]
$\mu$ & Vector natural death rate  \\[0.5 ex]
\bottomrule
\end{tabularx}
\label{table:variables} 
\end{table}

\begin{table}[h]
\caption{Definition of the linking parameters}
\label{table:variables2} 
\centering 
\begin{tabularx}{\textwidth}{>{} lX}
\toprule
Variable/Parameter  & Meaning \\ [0.5ex]
\toprule
\\
$b_0$ & the parasite cost coefficient, \\[0.5 ex]
$a_0$  & the transmission efficiency of the parasitic infection, \\[0.5 ex]
$b_1$ & the immune response cost coefficient, \\[0.5 ex]
$a_1$ &  half-saturation constant in transmission rate, \\[0.5 ex]
$c_0$ &  saturation constant in recovery rate, \\[0.5 ex]
$\epsilon_0$ &  half-saturation constant in recovery rate, \\[0.5 ex]
$d_1$ &  half-saturation constant of vector transmission rate, \\[0.5 ex]
$d_0$ & saturation constant of vector transmission rate, \\[0.5 ex]
\bottomrule
\end{tabularx}
\end{table}
\begin{figure}
\begin{tikzpicture}

\node[draw= black,rounded corners,text depth = 2cm,minimum width=15cm] (main){Vector Population};
\node[draw =black, ellipse,minimum width =3cm,minimum height=2.25cm,align=center] (sv) at ([xshift=2cm]main.west){};
\node[below right] (svtext) at ([xshift=-0.7cm,yshift=-0.2cm]sv.north) {$S_v(t)$};
\node [draw, diamond, aspect=2] (d1) at ([yshift =-0.3cm,xshift=-0.5cm]svtext.south) {};
\node [draw, diamond, aspect=2, right of =d1] (d2) {};
\node [draw, diamond, aspect=2,below of =d1,yshift=0.5cm] (d3) {};
\node [draw, diamond, aspect=2,right of =d3] (d4) {};
\node[draw =black, ellipse,minimum width =3cm,minimum height=2.25cm,align=center, right of =sv,node distance = 10cm] (Iv) {};
\node[below right] (Ivtext) at ([xshift=-0.7cm,yshift=-0.2cm]Iv.north) {$i_v(t,s)$};
\node [draw, diamond, fill =red!10, aspect=2] (dr1) at ([yshift =-0.3cm,xshift=-0.5cm]Ivtext.south) {};
\node [draw, diamond, fill =red!30,  aspect=2, right of =dr1] (dr2) {};
\node [draw, diamond, fill =red!55,  aspect=2,below of =dr1,yshift=0.5cm] (dr3) {};
\node [draw, diamond, fill =red!90,  aspect=2,right of =dr3] (dr4) {};

\node[draw= black,rounded corners,text depth = 2.5cm,minimum width=15cm, below of = main, node distance=3.5cm] (host){Host  Population};
\node[draw =black, ellipse,minimum width =3cm,minimum height=2.25cm,align=center] (sh) at ([xshift=2cm,yshift=-0.35cm]host.west){};
\node[below right] (shtext) at ([xshift=-0.7cm,yshift=-0.2cm]sh.north) {$S_H(t)$};
\node [draw, circle, minimum size=0.5cm ] (c1) at ([yshift =-0.3cm,xshift=-0.5cm]shtext.south) {};
\node [draw, circle, minimum size=0.5cm, right of =c1] (c2) {};
\node [draw, circle, minimum size=0.5cm, below of =c1,yshift=0.4cm] (c3) {};
\node [draw, circle ,minimum size=0.5cm, right of =c3] (c4) {};
\node[draw =black, ellipse,minimum width =3cm,minimum height=2.25cm,align=center, right of =sh,node distance = 5cm] (Ih) {};
\node[below right] (Ihtext) at ([xshift=-0.7cm,yshift=-0.2cm]Ih.north) {$i_H(t,\tau,s)$};
\node [draw, circle, fill=red!10, minimum size=0.5cm ] (cr1) at ([yshift =-0.3cm,xshift=-0.5cm]Ihtext.south) {};
\node [draw, circle, fill=red!30, minimum size=0.5cm, right of =cr1] (cr2) {};
\node [draw, circle, fill=red!55, minimum size=0.5cm, below of =cr1,yshift=0.4cm] (cr3) {};
\node [draw, circle, fill=red!90, minimum size=0.5cm, right of =cr3] (cr4) {};
\node[draw = black, ellipse,minimum width =3cm,minimum height=2.25cm,align=center, right of =Ih,node distance = 5cm] (rh) {};
\node[below right] (rhtext) at ([xshift=-0.7cm,yshift=-0.2cm]rh.north) {$R_H(t)$};
\node [draw, circle, fill=green!20, minimum size=0.5cm ] (cb1) at ([yshift =-0.3cm,xshift=-0.5cm]rhtext.south) {};
\node [draw, circle, fill=green!20, minimum size=0.5cm, right of =cb1] (cb2) {};
\node [draw, circle, fill=green!20, minimum size=0.5cm, below of =cb1,yshift=0.4cm] (cb3) {};
\node [draw, circle, fill=green!20, minimum size=0.5cm, right of =cb3] (cb4) {};

\node[draw,circle, minimum size =2.5cm, node distance =2.5cm, below of =cr3, xshift = -1cm](whd){};
\node[draw = red,fill=red, ellipse,star,star points=10, align=center] (pathogen) at ([yshift =-0.4cm]whd.north){};
\node[draw=none,fill=none,right of =pathogen, xshift =3.5cm] (pt) {Pathogen level in host ($P(\tau,s)$)};
\node[draw = none, fill =none, below of =pathogen, yshift =0.2cm] (antibody) {\includegraphics[scale=0.08]{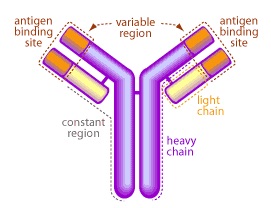}};
\node[draw=none,fill=none,right of =antibody, xshift =3.1cm](a1t)  {IgM level in host ($M(\tau,s)$)};
\node[draw = none, fill =none, below of =antibody, yshift =0.2cm] (antibody2) {\includegraphics[scale=0.05]{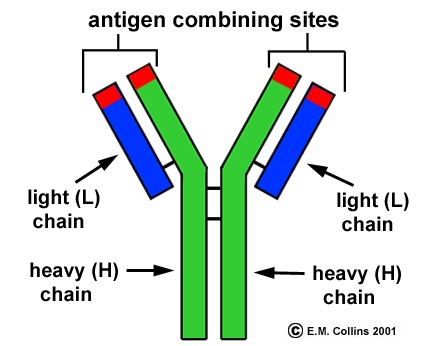}};
\node[draw=none,fill=none,right of =antibody2, xshift =3.1cm] (a2t) {IgG level in host ($G(\tau,s)$)};
\node[draw=none,fill=none,left of =whd, text width =  2.5cm, xshift =-2.5cm ]{Within Host Dynamics};

\draw[thick,->] ([xshift =0.3cm]pathogen.east) -- (pt);
\draw[thick,->] (antibody) -- (a1t);
\draw[thick,->] (antibody2) -- (a2t);
\draw (tangent cs:node=cr3,point={(whd.west)},solution=2) -- (tangent cs:node=whd,point={(cr3.west)});
\draw (tangent cs:node=cr3,point={(whd.east)},solution=1) -- (tangent cs:node=whd,point={(cr3.east)},solution=2);

\draw[thick,->] (sv.east) -- node[name=empty,yshift=0.2cm,text=black]{Transmission} (Iv.west);
\draw[->,dashed] (Ih.north) --  (empty);

\node [left of = sv,xshift=-1.8cm,yshift =-0.5cm](svout){};
\node [left of = sv,xshift=-1.8cm,yshift =0.5cm ](svin){};
\node [right of = Iv,xshift=1.8cm ](Ivout){};
\draw [thick,->, below] ([yshift =-0.5cm]sv.west)-- node[xshift=-0.2cm] {Deaths}(svout);
\draw [thick,->, above] (svin)-- node[xshift=-0.2cm] {Births}([yshift =0.5cm ]sv.west);
\draw [thick,->, below] (Iv.east)-- node[xshift=0.2cm]{Deaths}(Ivout);

\node [left of = sh,xshift=-1.8cm,yshift =-0.5cm](shout){};
\node [left of = sh,xshift=-1.8cm,yshift =0.5cm ](shin){};
\node [right of = Ih,xshift=1.8cm,yshift=1.2cm ](Ihout){};
\node [right of = rh,xshift=1.8cm](rhout){};
\draw [thick,->, below] ([yshift =-0.5cm]sh.west)-- node[xshift=-0.2cm] {Deaths}(shout);
\draw [thick,->, above] (shin)-- node[xshift=-0.2cm] {Births}([yshift =0.5cm ]sh.west);
\draw [thick,->, above] ([yshift=0.5cm]Ih.east)-- node[xshift=0.2cm]{Deaths}(Ihout);
\draw [thick,->, below] (rh.east)-- node[xshift=0.2cm]{Deaths}(rhout);
\draw[thick,->] (sh.east) -- node[name=empty2,yshift=0.2cm,text=black]{Transmission} (Ih.west);
\draw[->,dashed] ([yshift=-0.2cm]Iv.west) --  (empty2);
\draw[thick,->,below] (Ih.east) -- node{Recovery} (rh.west);
\end{tikzpicture}
\label{nestedframe}
\caption{Schematic illustration of the multi-scale model, structured by host and vector infection age.}
\end{figure}
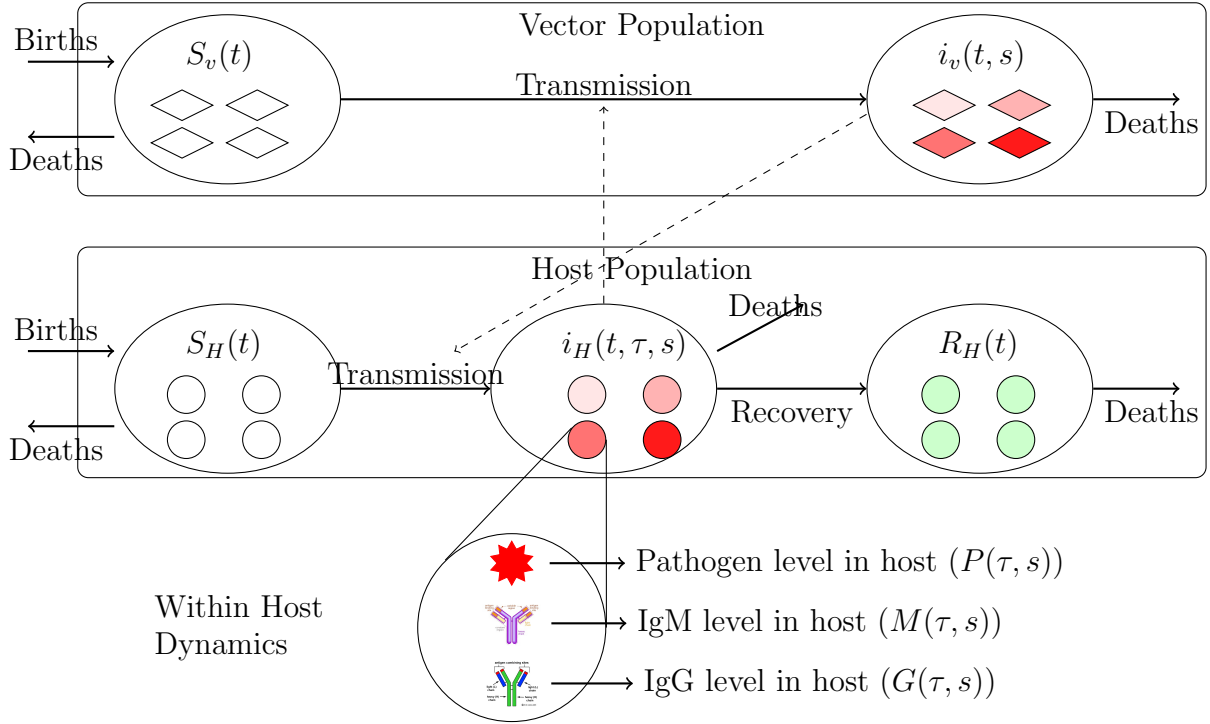

\noindent \emph{\textbf{\underline{The structured epidemic system:}}} 
To incorporate heterogeneity among vector to host inoculum size (across the vectors with different infection age), we formulate the infected host compartment as follows:

\vspace{-.5cm}
{\small \begin{align}\label{SIRhost_n}
\dfrac{\partial i_H(t,\tau,s)}{\partial t}+\dfrac{\partial i_H(t,\tau,s)}{\partial \tau} &=  -(\nu_H(\tau,s)+\gamma_H(\tau,s)+ d)i_H(t,\tau,s), \quad i_H(t,0,s) = S_H(t)\beta_v(s)i_v(t,s), 
\end{align}}
\vspace{-.5cm}

\noindent where $i_H(t,\tau,s)$ represents \emph{the density of hosts infected at time $t-\tau$ by a vector with infection age $s.$} In other words, $i_H(t,\tau,s)$ represents the density of hosts infected at time $t-\tau,$ with an infectious mosquito with vector to host inoculum size $P_0^{v\rightarrow h}$ (which is a function of viral density within-the infectious vectors; i.e. 
$P_0^{v\rightarrow h}=h(V(s, V_0)).$)

\noindent Furthermore the rates of change in the dynamics of susceptible ($S_H(t)$ ), and recovered ($R_H(t)$ ) host population size are described as follows:

\vspace{-.5cm}
{\small \begin{align*}
\dfrac{d S_H(t)}{dt} &= f(N(t)) -S_H(t)\int_0^{\infty} {\beta_v(s)i_v(t,s)ds}-d S_H(t), \quad
\dfrac{d R_H(t)}{dt} &\hspace{-4mm}= \int _0^{\infty}{\int _0^{\infty} {\gamma_H(\tau,s)i_H(t,\tau,s) ds} d\tau}-d R_H(t),
\end{align*}}
\vspace{-.5cm}

\noindent where parameters for the host population include: $f(N(t))$ the host recruitment rate with total host population size $N_H(t)=S_H(t)+\int_0^\infty\int_0^\infty{i_H(t,\tau,s)}d\tau ds+R_H(t)$, $d$ the natural death rate of host, $\beta_v$ the transmission rate of infection from infected vectors to hosts, $\gamma_H$ host recovery rate, and $\nu_H$ host disease-induced death rate. For simplicity, we consider the recruitment rate to be constant:$f(N(t))= \Lambda.$ \\

\noindent Vectors are the only mechanism transmitting the disease to susceptible hosts. The \emph{age-since-infection structured vector model} is given by:

\vspace{-.5cm}
{\small \begin{align}
\dfrac{d S_v}{dt} &=\eta -  S_v(t) \int_0^{\infty}\int_0^{\infty}\beta_H(\tau,s)i_H(t,\tau,s) ds d\tau-\mu S_v(t), \notag \\
\dfrac{\partial i_v(t,s)}{\partial t}+\dfrac{\partial i_v(t,s)}{\partial s} &=  -\mu i_v(t,s),\quad i_v(t,0)  = S_v(t) \int_0^{\infty}\int_0^{\infty}\beta_H(\tau,s)i_H(t,\tau,s) ds d\tau, \label{SIvector_n} 
\end{align}}
\vspace{-.4cm}

\noindent where $i_v(t,s)$ represents the density of infected vectors at time $t$ with infection age $s.$ The parameters related to vector dynamics are: $\eta$ the birth/recruitment rate of vectors, $\mu$ the natural death rate of vectors, $\beta_H$ the transmission rate of infection from infected hosts to vectors. In the system \eqref{SIRhost_n}-\eqref{SIvector_n}, a portion of the susceptible hosts move to the infected compartment with a rate $\int_0^{\infty} {\beta_v(s)i_v(t,s)ds}$ through  bites by infected vectors $i_v(t,s)$ with infection age $s$. 

\noindent \emph{\textbf{\underline{Linking within-vector dynamics:}}}
The epidemiological parameters $\beta_H(\tau,s)$, $\gamma_H(\tau,s)$, and $\nu_H(\tau,s)$ are formulated similar to previous studies \cite{gulbudak2017vector} (confirmed by data \cite{handel2015crossing, fraser2014virulence}) as follows:

\vspace{-.8\baselineskip}
{\small\begin{align} \label{linkfns}
\beta_H(\tau,s)=\dfrac{a_0 P(\tau,s)^2}{a_1+ P(\tau,s)^2}, \ \nu_H(\tau,s)=b_0P(\tau,s)+b_1M(\tau,s)P(\tau,s), \ \gamma_H(\tau,s)=\dfrac{c_0G(\tau,s)}{G(\tau,s)+\epsilon_0}e^{-P(\tau,s)}. 
\end{align}}

\vspace{-.8\baselineskip}
The data \cite{handel2015crossing, fraser2014virulence} suggests that the transmission rate $\beta_H(\tau,s)$ is a Holling type II function with respect to the within-host pathogen load $P(\tau,s),$ where within-host pathogen load depends on host and vector infection age: $P(\tau,s)=P(\tau,V(s,V_0)).$ The parameters $a_0$ and $a_1$ are transmission and half saturation constants, respectively. In addition, similar to previous study \cite{gulbudak2017vector}, we formulate the recovery rate as a function of immune response $G(\tau,s)$ and inversely related to the viral load as shown in \eqref{linkfns}, where  $c_0$ is transmission constant and $\epsilon_0>0$ is proportionality constant and a small number.  This formulation of recovery translates into low pathogen load with sufficient IgG memory antibodies to prevent subsequent rise in pathogen load. Thus recovery rate is a decreasing function of pathogen ($P$) and increasing function of IgG immune response ($G$).  Furthermore, disease induced death rate $\nu$ is simply formulated as a linear function of $P$ (death due to pathogen resource use), and a function of $M$ (death due to aggressive immune response).

\noindent In addition, we formulate  the transmission rate from an infectious vector to susceptible host ($\beta_v(s)$) as follows:  $ \beta_v(s)=\dfrac{d_0 V(s)}{d_1+V(s)},$ where the parameters $d_0, d_1$ are saturation and half saturation constants, respectively. \\

\noindent Note that the within-vector viral kinetics $V(s)$ affects both the vector to host transmission $\beta_v(s)$ and vector to host inoculum size $P_0^{v\rightarrow h}$, and the latter alters the within-host virus-immune dynamics, and in turn the host to vector transmission $\beta_H(\tau,s)$.  This chain of across-scale interactions effectively introduces feedback from both scales.  In contrast to previous attempts of incorporating feedback between scales \cite{gandolfi2015epidemic}, our approach is amenable to analysis and biologically relevant for vector-borne diseases.  Although the host to vector inoculum size $V_0^{h \rightarrow v}$ (amount of pathogen in blood meal) can affect within-vector kinetics, the within-host dynamics are more sensitive to vector to host inoculum size.  Thus we reserve the significant mathematical complexity of a two-way ``infinite-dimensional'' feedback for potential future modeling work, and introduce a "friendly" formulation in next section \eqref{two_way_feedback} to vary the host to vector inoculum size $V_0^{h \rightarrow v}.$ 
 

\section{Analytical results}\label{analytical}

\subsection{Basic properties of the system} We assume that all parameters of the model are non-negative. In addition to that, we also consider the immune model initial conditions to be nonnegative: $P_0=P_0^{v\rightarrow h}, \ M_0, \ G_0\geq0$. Through this article, it satisfies that  
\begin{align*}
\beta_H(.,.),\nu_H(.,.), \gamma_H(.,.)\in L^\infty(0,\infty)^2, \quad \beta_v(.) \in L^\infty(0,\infty)
\end{align*}
We introduce 
\begin{equation}
\pi_H(\tau,s) =e^{-\D\int_0^{\tau}{\left( \nu_H(u,s)+\gamma_H(u,s)+d \right)du}}, \text{ with } \tau, s \geq 0.
\end{equation}
\noindent Integrating the second equation of the system \eqref{SIRhost_n}-\eqref{SIvector_n} along the characteristic lines, we obtain
\begin{equation}\label{charac_equ}
i_H(t,\tau,s)=
\begin{cases}
\pi_H(\tau,s)i_H(t-\tau,0,s), & \text{ if } t > \tau \geq 0,\\
\dfrac{\pi_H(\tau,s)}{\pi_H(\tau-t,s)}i_H(0,\tau-t,s), & \text{ if }  \tau>t\geq 0,\\
\end{cases}
\end{equation}
where $\pi_H(\tau,s)$ can be interpreted as the probability of host (whom is bitten with a vector at infection age $s$) still being in the infected class at host infection age $\tau.$ 

\noindent First note that by the equation \eqref{charac_equ}, we  have $\lim_{\tau \rightarrow \infty}i_H(t,\tau,s)=0, \ \forall \ t \in [0,\infty).$ Then by integrating both side of the equation \eqref{SIRhost_n} with respect to both independent variables $s$ and $\tau,$ we obtain the following equation:
\begin{equation}\label{IH}
I_H^{\prime}(t)=S_H(t)\int_0^\infty{\beta_v(s)i_v(t,s)ds}-\int_0^\infty{\int_0^\infty{\left( \nu_H(\tau,s)+\gamma_H(\tau,s)+d \right)i_H(t,\tau,s)d\tau}ds}
\end{equation}
where $I_H(t)=\int_0^\infty \int_0^\infty{i_H(t,\tau,s)d\tau}ds.$ 
Similarly, integrating the last equation of the system \eqref{SIRhost_n}-\eqref{SIvector_n} along the characteristic lines, we obtain
\begin{equation}\label{charac_equ_vector}
i_v(t,s)=
\begin{cases}
\pi_v(s)i_v(t-s,0), & \text{ if } t > s \geq 0,\\
\dfrac{\pi_v(s)}{\pi_H(s-t)}i_v(0,s-t), & \text{ if }  s>t\geq 0,\\
\end{cases}
\end{equation}
where $\pi_v(s)$ can be interpreted as the probability of vector still being in the infected class at infection age $s.$ 
Then by similar argument above, we also obtain 
\begin{equation}\label{Iv}
I_v^{\prime}(t)=S_v(t) \int_0^{\infty}\int_0^{\infty}\beta_H(\tau, s)i_H(t,\tau,s) ds d\tau - \mu I_v(t)
\end{equation}
Upon plugging in the boundary and initial conditions of \eqref{SIRhost_n} and \eqref{SIvector_n} to \eqref{IH} and \eqref{Iv}, respectively, we obtain a system of integro-differential equations.   Utilizing contraction mapping arguments, similar to methods in \cite{gulbudak2014modeling,webb}, the existence of unique solutions to the coupled system \eqref{SIRhost_n}-\eqref{SIvector_n} can be shown, which remain non-negative for any time $t.$ 
After adding all equations in the system  \eqref{SIRhost_n}-\eqref{SIvector_n}, we obtain
\begin{align*}
N_H^{\prime}\leq \Lambda -d N_H, \text{ and } N^{\prime}_v =\eta -\mu N_v, \text{ implying }
\limsup\limits_{t\rightarrow \infty} N_H(t)\leq \dfrac{\Lambda}{d} \text{ and } \lim_{t\rightarrow \infty} N_v(t)=\dfrac{\eta}{\mu}.
\end{align*}
Since solutions $\phi(t):=(S_H(t),i_H(t,\tau,s),R_H(t), S_v(t), i_v(t,s))$ remain non-negative, we have 
\begin{align*}
0\leq \limsup\limits\limits\limits_{t\rightarrow \infty} S_H(t), I_H(t), R_H(t) \leq \limsup\limits_{t\rightarrow \infty} N_H(t)\leq \dfrac{\Lambda}{d} \text{ and }0\leq \limsup\limits_{t\rightarrow \infty} S_v(t), I_v(t) \leq lim_{t\rightarrow \infty} N_v(t)=\dfrac{\eta}{\mu}.
\end{align*}
 \noindent Thus  solutions remain bounded for all time $t$ and are attracted to a bounded set as $t\rightarrow\infty$.  Furthermore, solutions to the system  \eqref{SIRhost_n}-\eqref{SIvector_n} form a $C^0$-semigroup, denoted $\Psi(t)$, in the state space $X:=\mathbb R_+\times L_+^1(0,\infty)^2 \times \mathbb R_+\times L_+^1(0,\infty)$ \cite{thieme}.  In particular, for $x\in X$, where $\phi(t)=\Psi(t)x$ denotes solution with initial condition $x$ (written above), the following holds: 
 \begin{align}\Psi(t+s)x=\Psi(t)(\Psi(s)x). \label{semi}\end{align}\\


\noindent The long term behavior of the solutions is determined in part by the equilibria that are time-independent solutions of the system \eqref{SIRhost_n}-\eqref{SIvector_n}. The system \eqref{SIRhost_n}-\eqref{SIvector_n} has a DFE $\mathcal E_0=(S_H^0, 0, 0, S^0_v,0),$ where $S^0_H=\Lambda / d, S_v^0=\eta/ \mu.$\\

\noindent Define the reproduction number as follows: 
\vspace{-.1cm}
{\small \begin{align}       
\mathcal R_0 = S_H^0S_v^0 \D\int_0^\infty  \beta_v(s)\pi_v(s) \D\int_0^\infty \beta_H(\tau,s) e^{-\D\int_0^\tau ( \nu_H(u,s)+\gamma_H(u,s)+ d)du} d\tau ds.
\end{align}}
\vspace{-.1cm}

\noindent
The \emph{basic reproduction number} $\mathcal R_0$ keeps track of the number of secondary infectious hosts produced by one infected host during its infectious time period in an entirely susceptible host population. The term $ S_v^0 \D\int_0^\infty \beta_H(\tau,s) e^{-\D\int_0^\tau ( \nu_H(u,s)+\gamma_H(u,s)+ d)du} d\tau$ is the average number of secondary infectious vectors produced one infectious host (whom is bitten by an infectious vector at infection age $s$) during its lifespan in a wholly susceptible vector population. In addition, $e^{-d\tau}$ is the probability of host having survived to infection age $\tau.$

\begin{theorem}The DFE $\mathcal E_0$ is locally asymptotically stable if $\mathcal R_0<1$ and unstable if $\mathcal R_0>1.$
\end{theorem}
\begin{proof}To study the behavior of the solutions nearby an equilibrium, we first linearize the vector-host model ($\ref{SIRhost_n})-(\ref{SIvector_n}$) about the equilibrium $(S_H^+, \ \bar{i}_H(\tau,s),\ R_H^+,\ S_V^+, \ \bar{i}_v(s))$ by taking $S_H(t)= S_H^+ + x_H(t), i_H(t,\tau,s)= \bar{i}_H(\tau,s)+ y_H(t,\tau,s),  R_H(t)= R^+_H + z_H(t)$, $S_v(t)=S_v^++ x_v(t)$ and $i_v(t,s)=\bar{i}_v(s) + y_v(t,s).$ We look for eigenvalues of the linear operator - that is we look for solutions of the form $x_H(t)= \overline{x}_He^{\lambda t}, y_H(t,\tau,s)= \overline{y}_H(\tau,s)e^{\lambda t},z_H(t)= \overline{z}_He^{\lambda t}$, $x_v(t)=\overline{x}_v e^{\lambda t},$ and $y_v(t,s)= \overline{y}_v(s)e^{\lambda t}$,  where $\overline{x}_H,\ \overline{y}_H(\tau), \ \overline{z}_H, \ \overline{x}_v$ , \ $\overline{y}_v$ are arbitrary non-zero constants (a function of $\tau$ or $s$ in the case of $y_i$ for $i \in \{ H,v\} $). This process results in the following system (the bars have been omitted):

\small{\begin{equation}\label{linearized system}
\begin{cases}
\lambda x_H&\hspace{-4mm}=  - x_H\D\int_0^{\infty}{\beta_v(s)\bar{i}_v(s)ds}-S^+_H\D\int_0^{\infty}{\beta_v(s)y_v(s)ds}+\omega z_H-d x_H, \vspace{1.0mm}\\
\D\frac{d y_H(\tau,s)}{d \tau} +\lambda y_H(\tau,s)  
&\hspace{-4mm}=  
                 - \left( \nu_H(\tau,s)+\gamma_H(\tau,s)+d \right)y_H(\tau,s), \vspace{1.0mm}\\ 
y_H(0,s) &\hspace{-4mm} = x_H\beta_v(s)\bar{i}_v(s) + S^+_H \beta_v(s)y_v(s)  \vspace{1.0mm},\\ \\
\lambda z_H &\hspace{-4mm} =\D\int_0^{\infty}\D\int_0^{\infty}{\gamma_H(\tau,s)y_H(\tau,s)ds d\tau}-(d+\omega)z_H  \vspace{1.0mm},\\ 
\lambda x_v &\hspace{-4mm}=  -S_v^+ \D\int_0^\infty \D\int_0^{\infty}{\beta_H(\tau,s)y_H(\tau,s)ds d\tau} - x_v \D\int_0^\infty \D\int_0^\infty{\beta_H(\tau,s)\bar{i}_H(\tau,s)}ds d\tau -\mu x_v \vspace{1.0mm},\\
\D\frac{d y_v(s)}{ds}+\lambda y_v(s) &\hspace{-4mm}=-\mu y_v(s)   \vspace{1.0mm}.\\
y_v(0) &\hspace{-4mm} = S_v^+ \D\int_0^\infty \D\int_0^{\infty}{\beta_H(\tau,s)y_H(\tau,s)ds d\tau} + x_v \D\int_0^\infty \D\int_0^\infty{\beta_H(\tau,s)\bar{i}_H(\tau,s)}ds d\tau  \vspace{1.0mm},\\ \\ 
\end{cases}
\end{equation}}
\vskip-0.09in
\noindent 

\noindent Solutions of the system (\ref{linearized system}) give the eigenvectors and eigenvalues $\lambda$ of the differential operator.  As in \cite{martcheva2003progression},  it can be shown that knowing the distribution of the eigenvalues is sufficient to determine the stability of a given equilibrium for PDEs operators. In other words, as similar to ODEs, if all eigenvalues have negative real parts, the corresponding equilibrium is locally stable; if there is an eigenvalue with a positive real part, then the equilibrium is unstable. Because of that, we will concentrate on investigating eigenvalues.

\noindent Equilibrium of interest is DFE $\mathcal E_0=(S^0_H,0,0,S^0_v,0)$. Hence the system
(\ref{linearized system}) simplifies to the following  system:

\begin{equation}\label{linearized system at DFE}
\begin{cases}
\lambda x_H&\hspace{-4mm}= -S^0_H\D\int_0^{\infty}{\beta_v(s)y_v(s)ds}+\omega z_H-dx_H, \vspace{1.0mm}\\

\D\frac{d y_H(\tau,s)}{d \tau} +\lambda y_H(\tau,s)  
&\hspace{-4mm}=  
                 - \left( \nu_H(\tau,s)+\gamma_H(\tau,s)+d \right)y_H(\tau,s), \vspace{1.0mm}\\ 
                 
y_H(0,s) &\hspace{-4mm} =  S^0_H \beta_v(s)y_v(s)  \vspace{1.0mm},\\ \\

\lambda z_H &\hspace{-4mm} =\D\int_0^{\infty}\D\int_0^{\infty}{\gamma_H(\tau,s)y_H(\tau,s)ds d\tau}-(d+\omega)z_H  \vspace{1.0mm},\\ 

\lambda x_v &\hspace{-4mm}=  -S_v^0 \D\int_0^\infty \D\int_0^{\infty}{\beta_H(\tau,s)y_H(\tau,s)ds d\tau} -\mu x_v \vspace{1.0mm},\\

\D\frac{d y_v(s)}{ds}+\lambda y_v(s) &\hspace{-4mm}=-\mu y_v(s)   \vspace{1.0mm}.\\

y_v(0) &\hspace{-4mm} = S_v^0 \D\int_0^\infty \D\int_0^{\infty}{\beta_H(\tau,s)y_H(\tau,s)ds d\tau}   \vspace{1.0mm},\\ \\ 
\end{cases}
\end{equation}
\vskip-0.09in
\noindent 
Solving the differential equations in the linearized system \eqref{linearized system at DFE}, we obtain 
\begin{equation}\label{y_H}
y_H(\tau,s) = y_H(0,s)e^{-\lambda \tau}\pi_H(\tau,s), \textit{ where } \pi_H(\tau,s) =e^{-\int_0^{\tau}{\left( \nu_H(u,s)+\gamma_H(u,s)+d \right)du}},
\end{equation}
and
\begin{equation}\label{y_v}
y_v(s)=y_v(0)e^{-\lambda s}\pi_v(s),  \textit{ where } \pi_v(s)=e^{-\mu s}.
\end{equation}
Also by the boundary conditions in \eqref{linearized system at DFE}, we have 
\begin{equation}\label{y_0}
y_H(0,s)=S^0_H\beta_v(s)y_v(s), \ y_v(0)=S_v^0 \D\int_0^\infty \D\int_0^{\infty}{\beta_H(\tau,s)y_H(\tau,s)ds d\tau}.
\end{equation}
Substituting (\ref{y_H}) and \eqref{y_0} into (\ref{y_v}) and canceling $y_v(0)$, we get the following characteristic equation for $\lambda:$
\begin{equation}\label{chrac_eqn}
 1=S_H^0S_v^0 \D\int_0^\infty \D{ \beta_v(s)\pi_v(s) \int_0^\infty{\beta_H(\tau,s) \pi_H(\tau,s) e^{-\lambda(\tau+s)}d\tau}ds}.
\end{equation} 
The equation \eqref{chrac_eqn} is a trancedental equation; i.e. it involves trancedental functions. The above equation may have many solutions. To show stability of the DFE, we need to show that all solutions $\lambda$ of the above equation have negative real parts. If there is a solution $\lambda$ with positive real part, then the DFE is unstable. To investigate this, we denote by 
\begin{align*}
G (\lambda) =  S_H^0S_v^0 \D\int_0^\infty \D{ \beta_v(s)\pi_v(s) \int_0^\infty{\beta_H(\tau,s) \pi_H(\tau,s) e^{-\lambda(\tau+s)}d\tau}ds.}
\end{align*}
Notice that $G(0)=\mathcal R_0.$ If $\mathcal R_0>1$ and $\beta_H(\tau,s)$ is strictly positive on a positive interval, then the function $G(\lambda)$ is a decreasing function of $\lambda.$ Since $G(0)>1$ and $\lim_{\lambda \rightarrow \infty} G(\lambda)=0,$ then there exists $\lambda^+>0:\ G(\lambda^+)=1.$ Hence the DFE is unstable. \\
Now let $\mathcal R_0<1.$ Then for all $\lambda=a+ib$ with $a \geq 0,$ we have 

\begin{align*} 
\left| G(\lambda) \right| & \leq S_H^0S_v^0 \D\int_0^{\infty} \D \beta_v(s)\pi_v(s) \int_0^{\infty} \beta_H(\tau,s) \pi_H(\tau,s) \left| e^{-\lambda(\tau+s)} \right| d\tau ds  \\
&\leq S_H^0S_v^0 \D\int_0^{\infty} \D \beta_v(s)\pi_v(s) \int_0^{\infty} \beta_H(\tau,s) \pi_H(\tau,s) e^{-a(\tau+s)}d\tau ds  \leq \mathcal R_0 <1. 
\end{align*}

Then $\lambda$'s whose real part is non-negative can not satisfy the equation $G(\lambda)=1.$ Therefore the DFE is locally asymptotically stable in this case. 
\end{proof}

Under certain conditions, this result can be extended to global stability of $\mathcal E_0$ by means of Lyapunov functions.
\begin{theorem}  Suppose that 
$$\mathcal R_v:=S_H^0\int_0^{\infty}{\beta_v(s)\pi_v(s)ds}\leq 1, \quad \textit{and} \quad  
\mathcal R_H:=S_v^0 \left( \max\limits_{s \in [0,\infty)}\int{\beta_H(\tau,s)\pi_H(\tau,s)d\tau}\right)\leq 1.$$
Then the DFE $\mathcal E_0$ is globally asymptotically stable. 
\end{theorem}
\begin{proof} First define the positive functions
\begin{align*}
\alpha(s):=\int_s^{\infty}{\beta_v(l) e^{-\int_s^l{\mu dz}}dl,} \quad \omega(\tau,s):=\int_{\tau}^{\infty}{\beta_H(l,s) e^{-\int_\tau^l{\tilde{\delta}(z,s)dz}}dl},
\end{align*}
where $\tilde{\delta}(\tau,s):= \nu_H(\tau,s)+\gamma_H(\tau,s)+ d$. Then by Leibiniz rule, the derivatives of $\alpha(s)$ with respect to $s$ and $ \omega(\tau,s)$ with respect to $\tau$ satisfy 
\begin{align}\label{derivatives}
\alpha^\prime (s)=\dfrac{d \alpha(s)}{ds}=\mu \alpha(s)-\beta_v(s) \quad \textit{and} \quad 
\dfrac{\partial \omega(\tau,s)}{\partial \tau}=\tilde{\delta}(\tau,s)\omega(\tau,s)-\beta_H(\tau,s).
\end{align}
Let us consider any solution $(S_H(t),i_H(t,\tau,s), R_H(t),S_v(t), i_v(t,s))$  of the model \eqref{SIRhost_n}-\eqref{SIvector_n} with the non-negative initial data. We define a function  $W(t)$ as follows:
\begin{align*}
W=S_H^0g(\dfrac{S_H}{S_H^0})+S_v^0g(\dfrac{S_v}{S_v^0})+\dfrac{1}{\alpha(0)}\int_0^\infty{\alpha(s)i_v(t,s)ds}+\int_0^\infty{\dfrac{1}{\omega(0,s)} \int_0^\infty{\omega(\tau,s) i_H(t,\tau,s)ds}d\tau}
\end{align*}
where $g(x)= x-\ln(x)-1.$
First note that
\begin{align*}
&\dfrac{d}{dt}\left(\dfrac{1}{\alpha(0)}\int_0^\infty{\alpha(s)i_v(t,s)ds} \right)\\
&\qquad=\dfrac{1}{\alpha(0)}\left(\int_0^\infty{\alpha(s)\left( -\dfrac{\partial i_v(t,s)}{\partial s} -\mu i_v(t,s) \right)ds}\right)\\
&\qquad=\dfrac{1}{\alpha(0)}\left( -\int_0^\infty{\alpha(s)\mu i_v(t,s)ds}-\int_0^\infty{\alpha(s)\dfrac{\partial i_v(t,s)}{\partial s}ds}\right) \\
&\qquad\qquad\qquad \textit{(by integration by parts)}\\
&\qquad=\dfrac{1}{\alpha(0)}\left(
 -\int_0^\infty{\alpha(s)\mu i_v(t,s)ds}-\alpha(s) i_v(t,s)|_{s=0}^{\infty} +\int_0^\infty{\alpha^\prime(s) i_v(t,s)ds}\right)\\
 &\qquad\qquad\qquad \textit{(Note that $\alpha(s) i_v(t,s)|_{s=0}^{\infty}=- \alpha(0) i_v(t,0)$)}\\
 &\qquad=\dfrac{1}{\alpha(0)}\left(
 -\int_0^\infty{\alpha(s)\mu i_v(t,s)ds}+\alpha(0) S_v(t) \int_0^{\infty}\int_0^{\infty}\beta_H(\tau,s)i_H(t,\tau,s) ds d\tau
 \right.\\
 &\qquad\qquad\qquad \left.+\int_0^\infty{\alpha^\prime(s) i_v(t,s)ds}\right)\\
 &\qquad\qquad\qquad \textit{(Recall that $\alpha^\prime (s)=\mu \alpha(s)-\beta_v(s)$ by \eqref{derivatives})}\\
 &\qquad=\dfrac{1}{\alpha(0)}\left(
\alpha(0) S_v(t) \int_0^{\infty}\int_0^{\infty}\beta_H(\tau,s)i_H(t,\tau,s) ds d\tau-\int_0^\infty{\beta_v(s) i_v(t,s)ds}\right)\\
&\qquad=S_v(t) \int_0^{\infty}\int_0^{\infty}\beta_H(\tau,s)i_H(t,\tau,s) ds d\tau-\dfrac{1}{\alpha(0)}\int_0^\infty{\beta_v(s) i_v(t,s)ds},
\end{align*}
\begin{align*}
&\dfrac{d}{dt}\left(\int_0^\infty{\dfrac{1}{\omega(0,s)}\int_0^\infty{\omega(\tau,s)i_H(t,\tau,s)d\tau}ds} \right)\\
&\qquad=\int_0^\infty{\dfrac{1}{\omega(0,s)}\int_0^\infty{ \omega(\tau,s)\dfrac{\partial i_H(t,\tau,s)}{\partial t}  d\tau}ds} \\
&\qquad=-\int_0^\infty{\dfrac{1}{\omega(0,s)}\int_0^\infty{ \omega(\tau,s) \left( \dfrac{\partial i_H(t,\tau,s)}{\partial \tau}+\tilde{\delta}(\tau,s)i_H(t,\tau,s) \right) d\tau}ds}\\
&\qquad=-\int_0^\infty{\dfrac{1}{\omega(0,s)} \int_0^\infty{ \omega(\tau,s) \dfrac{\partial i_H(t,\tau,s)}{\partial \tau}d\tau} ds}
-\int_0^\infty{\dfrac{1}{\omega(0,s)} \int_0^\infty{ \omega(\tau,s)\left(\tilde{\delta}(\tau,s)i_H(t,\tau,s) \right) d\tau}ds}\\
&\qquad=-\int_0^\infty{\dfrac{1}{\omega(0,s)}  \left(\omega(\tau,s)i_H(t,\tau,s)|_{\tau=0}^{\infty}-\int_0^\infty{\dfrac{\partial \omega(\tau,s)}{\partial \tau} i_H(t,\tau,s)d\tau}\right)ds} \\
&\qquad\qquad\qquad -\int_0^\infty{\dfrac{1}{\omega(0,s)} \int_0^\infty{ \omega(\tau,s)\left(\tilde{\delta}(\tau,s)i_H(t,\tau,s) \right) d\tau}ds}\\
 &\qquad\qquad\qquad \textit{(Recall that $\dfrac{\partial \omega(\tau,s)}{\partial \tau}=\tilde{\delta}(\tau,s)\omega(\tau,s)-\beta_H(\tau,s)$ by \eqref{derivatives})}\\
&\qquad=-\int_0^\infty{\dfrac{1}{\omega(0,s)}  \left(-\omega(0,s)i_H(t,0,s)-\int_0^\infty{\left(\tilde{\delta}(\tau,s)\omega(\tau,s)-\beta_H(\tau,s) \right) i_H(t,\tau,s)d\tau}\right)ds} \\
&\qquad\qquad\qquad -\int_0^\infty{\dfrac{1}{\omega(0,s)} \int_0^\infty{ \omega(\tau,s)\left(\tilde{\delta}(\tau,s)i_H(t,\tau,s) \right) d\tau}ds}\\
&\qquad=\int_0^\infty{i_H(t,0,s)}ds+\int_0^\infty{\dfrac{1}{\omega(0,s)}\left( \int_0^\infty{\left(\tilde{\delta}(\tau,s)\omega(\tau,s)-\beta_H(\tau,s) \right) i_H(t,\tau,s)d\tau}\right)ds} \\
&\qquad\qquad\qquad -\int_0^\infty{\dfrac{1}{\omega(0,s)} \int_0^\infty{ \omega(\tau,s)\left(\tilde{\delta}(\tau,s)i_H(t,\tau,s) \right) d\tau}ds}\\
&\qquad=\int_0^\infty{i_H(t,0,s)}ds-\int_0^\infty{\dfrac{1}{\omega(0,s)}\left( \int_0^\infty{\beta_H(\tau,s) i_H(t,\tau,s)d\tau}\right)ds} \\
&\qquad=\int_0^\infty{S_H(t)\beta_v(s)i_v(t,s)}ds-\int_0^\infty{\dfrac{1}{\omega(0,s)}\left( \int_0^\infty{\beta_H(\tau,s) i_H(t,\tau,s)d\tau}\right)ds}.
\end{align*}
Therefore the derivative of $W$ along solutions is:
\begin{align*}
\dot{W}=&\dfrac{d}{dt}\left(S_H^0g(\dfrac{S_H}{S_H^0})+S_v^0g(\dfrac{S_v}{S_v^0})\right)+\dfrac{d}{dt}\left(\dfrac{1}{\alpha(0)}\int_0^\infty{\alpha(s)i_v(t,s)ds}\right)\\
  &\qquad\qquad+\dfrac{d}{dt}\left(\int_0^\infty{\dfrac{1}{\omega(0,s)} \int_0^\infty{\omega(\tau,s) i_H(t,\tau,s)ds}d\tau}\right)\\
  &= \left(1-\dfrac{S_H^0}{S_H} \right)(S_H^\prime)+ \left(1-\dfrac{S_v^0}{S_v} \right)(S_v^\prime)+\dfrac{d}{dt}\left(\dfrac{1}{\alpha(0)}\int_0^\infty{\alpha(s)i_v(t,s)ds}\right)\\
  &\qquad\qquad+\dfrac{d}{dt}\left(\int_0^\infty{\dfrac{1}{\omega(0,s)} \int_0^\infty{\omega(\tau,s) i_H(t,\tau,s)ds}d\tau}\right)\\ 
   &= \left(1-\dfrac{S_H^0}{S_H} \right)\left(\Lambda -S_H(t)\int_0^{\infty} {\beta_v(s)i_v(t,s)ds}-d S_H(t)\right)\\
     &\qquad\qquad+ \left(1-\dfrac{S_v^0}{S_v} \right)\left(\eta -  S_v(t) \int_0^{\infty}\int_0^{\infty}\beta_H(\tau,s)i_H(t,\tau,s) ds d\tau-\mu S_v(t)\right)\\
    &\qquad\qquad +S_v(t) \int_0^{\infty}\int_0^{\infty}\beta_H(\tau,s)i_H(t,\tau,s) ds d\tau-\dfrac{1}{\alpha(0)}\int_0^\infty{\beta_v(s) i_v(t,s)ds}\\
  &\qquad\qquad+\int_0^\infty{S_H(t)\beta_v(s)i_v(t,s)}ds-\int_0^\infty{\dfrac{1}{\omega(0,s)}\left( \int_0^\infty{\beta_H(\tau,s) i_H(t,\tau,s)d\tau}\right)ds}\\  
&=\left(\Lambda -S_H(t)\int_0^{\infty} {\beta_v(s)i_v(t,s)ds}-d S_H(t)\right)-\dfrac{S_H^0}{S_H} \left(\Lambda -S_H(t)\int_0^{\infty} {\beta_v(s)i_v(t,s)ds}-d S_H(t)\right)\\
     &\qquad\qquad+ \left(1-\dfrac{S_v^0}{S_v} \right)\left(\eta -  S_v(t) \int_0^{\infty}\int_0^{\infty}\beta_H(\tau,s)i_H(t,\tau,s) ds d\tau-\mu S_v(t)\right)\\
    &\qquad\qquad +S_v(t) \int_0^{\infty}\int_0^{\infty}\beta_H(\tau,s)i_H(t,\tau,s) ds d\tau-\dfrac{1}{\int_0^{\infty}{\beta_v(l) \pi_v(l)dl}}\int_0^\infty{\beta_v(s) i_v(t,s)ds}\\
  &\qquad\qquad+\int_0^\infty{S_H(t)\beta_v(s)i_v(t,s)}ds-\int_0^\infty{\dfrac{1}{\int_{0}^{\infty}{\beta_H(l,s) \pi_H(l,s)}dl}\left( \int_0^\infty{\beta_H(\tau,s) i_H(t,\tau,s)d\tau}\right)ds}\\  
  &=\left(\Lambda -d S_H(t)\right)-\dfrac{S_H^0}{S_H} \left(\Lambda -S_H(t)\int_0^{\infty} {\beta_v(s)i_v(t,s)ds}-d S_H(t)\right)\\
     &\qquad\qquad+ \left(\eta -\mu S_v(t)\right)-\dfrac{S_v^0}{S_v} \left(\eta -  S_v(t) \int_0^{\infty}\int_0^{\infty}\beta_H(\tau,s)i_H(t,\tau,s) ds d\tau-\mu S_v(t)\right)\\
    &\qquad\qquad -\dfrac{1}{\int_0^{\infty}{\beta_v(l) \pi_v(l)dl}}\int_0^\infty{\beta_v(s) i_v(t,s)ds}\\
  &\qquad\qquad-\int_0^\infty{\dfrac{1}{\int_{0}^{\infty}{\beta_H(l,s) \pi_H(l,s)}dl}\left( \int_0^\infty{\beta_H(\tau,s) i_H(t,\tau,s)d\tau}\right)ds}\\  
   &=-\dfrac{d}{S_H(t)}(S_H^0 -S_H(t))^2 - \dfrac{\mu}{S_v(t)}(S_v^0-S_v(t))^2\\
    &\qquad\qquad +(1-\dfrac{1}{S_H^0\int_0^{\infty}{\beta_v(l) \pi_v(l)dl}})S_H^0\int_0^\infty{\beta_v(s) i_v(t,s)ds}\\
  &\qquad\qquad+\int_0^\infty{(1-\dfrac{1}{S_v^0\int_{0}^{\infty}{\beta_H(l,s) \pi_H(l,s)}dl})\left(S_v^0 \int_0^\infty{\beta_H(\tau,s) i_H(t,\tau,s)d\tau}\right)ds}     
  \\  
   &\leq-\dfrac{d}{S_H(t)}(S_H^0 -S_H(t))^2 - \dfrac{\mu}{S_v(t)}(S_v^0-S_v(t))^2+S_H^0\left(1-\dfrac{1}{\mathcal R_v}\right)\int_0^\infty{\beta_v(s) i_v(t,s)ds}\\
  &\qquad\qquad+S_v^0\left(1-\dfrac{1}{\mathcal R_H}\right)\int_0^\infty\int_0^\infty \beta_H(\tau,s) i_H(t,\tau,s)d\tau ds    
  \end{align*}
Therefore $\mathcal R_v\leq 1$ and $\mathcal R_H\leq 1$ ensures that $\dot W\leq 0$ holds.  Note that the within-host infection eventually clears (Theorem \ref{withhost}), implying that there is a finite maximum age of host-vector transmission.  Also the within-vector viral load converges to equilibrium, implying that transmission rate $\beta_v(s)$ and inoculum $h(V(s,V_0^{h \rightarrow v}))$ are eventually constant, allowing us essentially to separate solutions into a part with variable rates dependent on $s$ for $s<s_M$ and constant rates for $s>s_M$.  These two features ensure a finite maximum age (so that all forward paths have compact closure), and allow us to apply Lyapunov-Lasalle Invariance Principle for functional differential equations \cite{hale1963stability}.  Thus solutions tend to the largest invariant set, $\mathcal A$, where $\dot W=0$.  Equality requires that $S_H(t)=S_H^0$ and $S_v(t)=S_v^0$.  If $\mathcal R_v<1$ or $\mathcal R_H< 1$, then either $i_H(t,\tau,s)\equiv 0$ or $i_v(t,s)\equiv 0$ on account of characteristic solutions \eqref{charac_equ}-\eqref{charac_equ_vector} and invariance of $\mathcal A$.  This readily implies $i_H(t,\tau,s)\equiv 0$ and $i_v(t,s)\equiv 0$.  If $\mathcal R_v= 1$ and $\mathcal R_H=1$, from $S_H'$ and $S_v'$ equations, we still obtain that $\int_0^\infty{\beta_v(s) i_v(t,s)ds}=0$ and $\int_0^\infty\int_0^\infty \beta_H(\tau,s) i_H(t,\tau,s)d\tau ds=0$.  So the same argument implies $i_H(t,\tau,s)\equiv 0$ and $i_v(t,s)\equiv 0$.  Therefore the DFE $\mathcal E_0$ is globally asymptotically stable.
\end{proof}
We remark that $\mathcal R_0\leq \mathcal R_v \mathcal R_H$ by H\"{o}lder's inequality.  The interpretation of $\mathcal R_v$ is the average secondary host transmissions due to an infected vector, and for $\mathcal R_H$, the expected secondary vector transmissions by an infected host maximized over all possible vector to host inoculum sizes, $V(s,V_0^{h \rightarrow v})$.  The result guarantees when both are less than unity, the disease eradicates.

\begin{proposition} Assume that $V(s,V_0^{h \rightarrow v})=V_0^{h \rightarrow v}.$  Then if $\mathcal R_0<1,$ the DFE $\mathcal E_0$ is globally asymptotically stable. 
\end{proposition}
 \begin{proof}
 \noindent Recall the solutions of the system \eqref{SIRhost_n}-\eqref{SIvector_n} obtained along the characteristic lines: 
\begin{equation}\label{SIRsol}
i_H(t,\tau,s) =
\begin{cases}
i_H(t-\tau,0,s) \pi_H(\tau,s)  &, t>\tau\\
i_H^0(\tau-t,s) \displaystyle\frac{\pi_H(\tau,s)}{\pi_H(\tau-t,s)} &, t<\tau
\end{cases}
\end{equation}
\noindent Substituting (\ref{SIRsol}) in the first equation in (\ref{SIvector_n}), and noting that $S_V(t)\le S_v^0$ and $S_H(t)\le S_H^0,$ we obtain
\begin{equation}\label{ineq1}
\begin{array}{l}
\displaystyle \frac{d I_V}{d t} = S_V(t)\D\int_0^\infty \D\int_0^t \beta_H(\tau,s) i_H(t-\tau,0,s)
\pi_H(\tau,s) \, d\tau ds \\
\qquad\qquad\qquad\qquad+ S_V(t) \D\int_0^{\infty} \D\int_t^{\infty} \beta_H(\tau,s)
i_H^0(\tau-t,s)\displaystyle\frac{\pi
_H(\tau,s)}{\pi_H(\tau-t,s)}\,d\tau ds -\mu I_V(t)\\
\qquad \le\displaystyle S_H^0 S_v^0\D\int_0^\infty \D\int_0^\infty \beta_v(s)i_v(t,s) \beta_H(\tau,s)\pi_H(\tau,s) d\tau ds +\underbrace{S_v^0 \bar{\beta}e^{-d t}
\int_0^\infty \int_0^\infty i_H^0(\tau,s) d\tau ds}_{\mathcal{O}{(e^{-dt})}} -\mu I_V(t)\\
\quad \textit{(Note that $V(s,V_0^{h \rightarrow v})=V_0^{h \rightarrow v} \Rightarrow \beta_v(s)=\tilde{\beta}_v$ for some constant $\tilde{\beta}_v=\dfrac{d_0 V_0^{h \rightarrow v}}{d_1+V_0^{h \rightarrow v}}>0$} \\
\quad \textit{ and $\beta_H(\tau,s)=\tilde{\beta}_H(\tau), \pi_H(\tau,s)=\tilde{\pi}_H(\tau) \ \forall s\in [0, \infty)$ )}\\
\qquad=S_H^0 S_v^0 \tilde{\beta}_v\D\int_0^\infty i_v(t,s)ds  \D\int_0^\infty\tilde{\beta}_H(\tau)\tilde{\pi}_H(\tau) d\tau +\underbrace{S_v^0 \bar{\beta}e^{-d t}
\int_0^\infty \int_0^\infty i_H^0(\tau,s) d\tau ds}_{\mathcal{O}{(e^{-dt})}} -\mu I_v(t)\\
\qquad=S_H^0 S_v^0 \tilde{\beta}_v I_v(t)  \D\int_0^\infty\tilde{\beta}_H(\tau)\tilde{\pi}_H(\tau) d\tau +\underbrace{S_v^0 \bar{\beta}e^{-d t}
\int_0^\infty \int_0^\infty i_H^0(\tau,s) d\tau ds}_{\mathcal{O}{(e^{-dt})}} -\mu I_V(t)\\
\qquad \le\displaystyle 
\mu \mathcal R_0 I_v(t)+\mathcal{O}{(e^{-dt})} -\mu I_v(t)=\mu I_v(t)(\mathcal R_0-1)+\mathcal{O}{(e^{-dt})} 
\\ \\
\end{array}
\end{equation}
\noindent Now, define $ \limsup\limits_t I_v :=I_v^\infty $.
Then by Fluctuation Lemma \cite{},  $\exists \{ t_n\}: I_v'(t_n)\to 0$ and $I_v(t_n) \rightarrow I_v^\infty $ as $t \rightarrow \infty,$ implying that $0\leq (\mathcal R_0-1)\mu I_v^\infty.$ 
Since $\mathcal R_0<1$, this implies that $ I_v ^\infty=0$. Furthermore, since $S_H(t)$ is bounded,  $\lim \limits_{t \to \infty} i_H(t,0,s) =0$. Similar
inequalities as the ones in (\ref{SIRsol}) imply that 
$\lim \limits_{t \to \infty} I_H(t) =0$. Therefore, by the differential equation in  the  system \eqref{SIRhost_n}-\eqref{SIvector_n},  $\lim \limits_{t \to \infty} R_H(t) =0.$ 
Since $\lim \limits_{t \to \infty} N = S^0_H$, that implies that $\lim \limits_{t \to \infty} S_H(t)
=S^0_H$. Similar reasoning applies to the vector population,
resulted in $\lim \limits_{t \to \infty} S_v(t) = S^0_v.$ This completes the proof.
\end{proof}

\begin{proposition} If $\mathcal R_0>1,$ there exists a unique positive endemic equilibrium $$\mathcal E^+ =\left(S^+_H,\ I^+_H=\D\int_0^\infty \D\int_0^\infty { \bar{i}_H(\tau,s)ds d\tau},\ R^+_H,\  S^+_v, \ I^+_v=\D\int_0^\infty {\bar{i}_v(s)ds}\right),$$ with components
\begin{align*}\label{eq_comp}
S^+_H &=\dfrac{S^0_H (\mathcal R_0-1)}{\mu S^0_v\mathcal R_0} 
\left[1/ \left(\dfrac{1}{\eta}+  \dfrac{\mathcal R_0}{S^0_H S^0_v}+\dfrac{1}{d}(1+ \D\int_0^\infty  \beta_v(s) \pi_v(s)  \D\int_0^\infty \nu_H(\tau,s)\pi_H(\tau,s) d\tau ds ))\right)\right]
+ \dfrac{S^0_H}{\mathcal R_0} , \vspace{1mm}\\ \vspace{1mm}
I^+_H&=S^+_H \bar{i}_v(0)\D\int_0^\infty \beta_v(s)\pi_v(s)\D\int_0^\infty \pi_H(\tau,s)d\tau ds, \quad R^+_H =\dfrac{S^+_H \bar{i}_v(0)}{d} \D\int_0^\infty \beta_v(s)\pi_v(s) \D\int_0^\infty \gamma_H(\tau,s) \pi_H(\tau,s)d\tau ds, \vspace{1mm}\\
S^+_v &= \frac{S^0_HS^0_v}{S^+_H \mathcal R_0},\quad I^+_v=  \bar{i}_v(0)\D\int_0^\infty \pi_v(s) ds,\vspace{1mm}\\
 \end{align*} 
and \begin{align*} 
\bar{i}_v(0)= \left( \dfrac{ S^0_H}{S^+_H}-\dfrac{S^0_H}{S^+_H\mathcal R_0}\right) / \left[ \dfrac{1}{\eta}+   \D\int_0^\infty  \beta_v(s) \pi_v(s)  \D\int_0^\infty \pi_H(\tau,s) d\tau ds  \right. \\
 \left. \qquad\qquad\qquad\qquad  +\dfrac{1}{d}(1+\D\int_0^\infty  \beta_v(s) \pi_v(s)  \D\int_0^\infty \nu_H(\tau,s)\pi_H(\tau,s) d\tau ds) \right].
\end{align*}
\end{proposition}
\begin{proof}
To find the endemic equilibria, we look for time-independent solutions with at least one non-zero infected compartment, which satisfy the system \eqref{SIRhost_n}-\eqref{SIvector_n} with the time derivatives equal to zero:
\begin{equation}\label{vectorhostequ.}
\begin{cases}
0 &\hspace{-4mm}= \Lambda - S^+_H \D\int_0^\infty \beta_v(s)\bar{i}_v(s)ds   -d S^+_H, \  \D\frac{\partial  \bar{i}_H(\tau,s)}{\partial \tau} =  
     -(\nu_H (\tau,s)+ \gamma_H(\tau,s)+ d)\bar{i}_H(\tau,s) \vspace{1.0mm},\\
\bar{i}_H(0,s) &\hspace{-4mm} = S^+_H  \beta_v(s)\bar{i}_v(s) , \  0 =\D\int_0^\infty \D\int _0^\infty {\gamma_H(\tau,s)\bar{i}_H(\tau,s)ds d\tau}- d R^+_H  \vspace{1.0mm},\\ 

0 &\hspace{-4mm}= \eta -  S^+_v \D\int_0^\infty \D\int_0^\infty {\beta_H(\tau,s) \bar{i}_H(\tau,s)ds d \tau}-\mu S^+_v , \ \D\frac{d \bar{i}_v(s)}{ds} =  
     -\mu(s) \bar{i}_v(s) \vspace{1.0mm},\\
\bar{i}_v(0) &\hspace{-4mm} = S^+_v\D\int_0^\infty \D\int_0^\infty \beta_H(\tau,s) \bar{i}_H(\tau,s) ds d\tau \vspace{1.0mm},\\ 
\end{cases}
\end{equation}
\vskip-0.09in
\noindent
An endemic equilibrium will be given by a non-trivial solution $(S_H^+, \ \bar{i}_H(\tau,s), \ R_H^+, \ S_v^+,\bar{i}_v(s)).$ We first solve the differential equations in the system \eqref{vectorhostequ.} and obtain the following implicit solutions: 
\begin{align}\label{iheq}
\bar{i}_H(\tau,s) = \bar{i}_H(0,s)\pi_H(\tau,s)= S^+_H  \beta_v(s)\bar{i}_v(s) \pi_H(\tau,s), \  \bar{i}_v(s)= \bar{i}_v(0)\pi_v(s)
\end{align}
Substituting this expression into the vector boundary condition in (\ref{vectorhostequ.}) and canceling $\bar{i}_v(0)$, we obtain $1= S^+_H S^+_v \D\int_0^\infty \beta_v(s)\pi_v(s) \D\int_0^\infty \beta_H(\tau,s)\pi_H(\tau,s)d\tau ds.$ Then the susceptible host equilibrium is:
\begin{align}
\label{bwoc1}
S^+_H= 1\ / \left(S^+_v \D\int_0^\infty \beta_v(s)\pi_v(s) \D\int_0^\infty \beta_H(\tau)\pi_H(\tau)d\tau ds \right).
\end{align}
From the third equation in (\ref{vectorhostequ.}), we can express $R^+_H$ in the terms of $\bar{i}_v(0):$
\begin{equation}\label{R_H}
R^+_H=\frac{S^+_H \bar{i}_v(0)}{d} \Gamma, \text{ where } \Gamma= \D\int_0^\infty \beta_v(s)\pi_v(s) \D\int_0^\infty \gamma_H(\tau,s) \pi_H(\tau,s)d\tau ds.
\end{equation}
Integrating the second differential equation in (\ref{vectorhostequ.}), we obtain
\begin{align}
\label{I_H'}
0= S^+_H\beta_v(s)\bar{i}_v(s) - \D\int_0^\infty{\left( \nu_H(\tau,s)+\gamma_H(\tau,s) \right)\bar{i}_H(\tau,s)d\tau ds}-d\D\int_0^\infty{\bar{i}_H(\tau,s)}d\tau,
\end{align}
where $\lim_{\tau \rightarrow \infty} \bar{i}_H(\tau,s)=0.$
Adding this equation to the first and the fourth equations in \eqref{vectorhostequ.}, we obtain the population size for the host at equilibrium as follows:
\begin{align} \label{N_H}
 N^+_H =\left(\Lambda-\D\int_0^\infty \D\int_0^\infty \nu_H(\tau,s) \bar{i}_H(\tau,s)d\tau ds\right) \ /d.
\end{align}
Next, we substitute \eqref{N_H} in the equilibrium for the total population size $N^+_H = S^+_H+I^+_H +R^+_H,$ and obtain
\begin{align}\label{sm}
\left(\Lambda-\D\int_0^\infty \D\int_0^\infty \nu_H(\tau,s) \bar{i}_H(\tau,s)d\tau ds\right)\ /d &= 1\ / \left(S^+_v \D\int_0^\infty \beta_v(s)\pi_v(s) \D\int_0^\infty \beta_H(\tau,s)\pi_H(\tau,s)d\tau ds\right) \\
&+ \D\int_0^\infty \D\int_0^\infty \bar{i}_H(0,s)\pi_H(\tau,s)d\tau ds+ \frac{S^+_H \bar{i}_v(0)}{d} \Gamma. \nonumber\\ \nonumber
\end{align}
Next we will solve the equation \eqref{sm} for $\bar{i}_v(0).$ Notice that by having an explicit expression for $\bar{i}_v(0),$ we can obtain $\bar{i}_H(\tau,s).$ From the equation \eqref{sm}, we have 
\begin{align}
\label{i_H}
S^0_H\left( 1-d \ / (\Lambda S^+_v \D\int_0^\infty \beta_v(s)\pi_v(s) \D\int_0^\infty \beta_H(\tau,s)\pi_H(\tau,s)d\tau ds ) \right) =  & \D\int_0^\infty \D\int_0^\infty \bar{i}_H(0,s)\pi_H(\tau,s)d\tau ds   \\
&\hspace{-20mm} + \left(S^+_H \bar{i}_v(0)+\D\int_0^\infty \D\int_0^\infty \nu_H(\tau,s) \bar{i}_H(\tau,s)d\tau ds \right) / d \nonumber \\ \nonumber
\end{align}
By the fourth equation in (\ref{vectorhostequ.}), we have the susceptible vector equilibrium as follows:
\begin{align}
\label{S_V}
S^+_v=\eta \ /  \left( \D\int_0^\infty \D\int_0^\infty {\beta_H(\tau,s) \bar{i}_H(\tau,s)ds d \tau}+\mu \right)
\end{align}
Substituting (\ref{S_V}) into the equation (\ref{i_H}) and rearranging it, we obtain 
\begin{align}\label{i_H111}
  \dfrac{S^0_H} {S^+_H \bar{i}_v(0)}(1-\frac{1}{\mathcal R_0}) &=  \dfrac{1}{\eta}
+  \D\int_0^\infty \beta_v(s) \pi_v(s) \D\int_0^\infty \pi_H(\tau,s) d\tau ds  \\
&\qquad\qquad + \dfrac{1}{d}(1+\D\int_0^\infty  \beta_v(s) \pi_v(s)  \D\int_0^\infty \nu_H(\tau,s)\pi_H(\tau,s) d\tau ds) .\nonumber
\end{align}
Rearranging the equation \eqref{i_H111}, we get
\begin{align*} 
\bar{i}_v(0)= \left( \dfrac{ S^0_H}{S^+_H}-\dfrac{S^0_H}{S^+_H\mathcal R_0}\right) / \left[ \dfrac{1}{\eta}+   \D\int_0^\infty  \beta_v(s) \pi_v(s)  \D\int_0^\infty \pi_H(\tau,s) d\tau ds  \right. \\
 \left. \qquad\qquad\qquad\qquad  +\dfrac{1}{d}(1+\D\int_0^\infty  \beta_v(s) \pi_v(s)  \D\int_0^\infty \nu_H(\tau,s)\pi_H(\tau,s) d\tau ds) \right]. \nonumber
\end{align*}
Therefore whenever $\mathcal R_0>1,$  $\bar{i}_v(0)$ is positive, establishing the result.
\end{proof}

\begin{proposition}Assume that within-vector viral load is constant through vector infection period; i.e. $V(s,V_0^{h \rightarrow v})=V_0^{h \rightarrow v}.$ Then if $\mathcal R_0>1,$ then $\mathcal E^+$ is locally asymptotically stable whenever it exists.
\end{proposition}
\begin{proof}
Consider the linearized system (\ref{linearized system}), where the equilibrium of interest is the endemic equilibrium. We eliminate $x_H(t), \ y_H(t,\tau,s), \ x_v(t), \ y_v(t,s)$ so that an equation in $\lambda$ is obtained. The first equation in (\ref{linearized system}) is an equality in the terms of $x_H(t), \ y_v(t,s)$ and $\lambda.$ Also by the second equality and boundary condition in the system (\ref{linearized system}), we have $y_H(t,\tau,s)= y_H(0,s,t)e^{-\lambda \tau}\pi_H(\tau,s),$ where $y_H(0,t,s) =  x_H(t)\beta_v(s)\bar{i}_v(s) + S^+_H \beta_v(s)y_v(t,s).$
Substituting this equality in the last equation in (\ref{linearized system}), we obtain an equality in the terms of $x_H, \ y_v,\ x_v$ and $\lambda.$ By the third equation in the system (\ref{linearized system}), we have 
$$x_v=-S_v^+ \D\int_0^\infty \D\int_0^{\infty}{\beta_H(\tau,s)y_H(t,\tau,s)ds d\tau} / \left(\lambda+ \D\int_0^\infty{\beta_H(\tau,s)\bar{i}_H(\tau,s)d\tau ds}+\mu \right),$$
which is in terms of $y_H(\tau,s),$ and $\lambda.$ Then substituting this equality also in the last equation in (\ref{linearized system}), we obtain the last equality in the terms of $x_H, \ y_v(0,t)$ and $\lambda.$  We obtain the following  system:
\begin{equation}\label{linerized system1}
\footnotesize{\begin{cases}
0&\hspace{-150mm}=\lambda x_H + y_v(t,0)S_H^+\D\int_0^\infty{\beta_v(s) e^{-\lambda s}\pi_v(s)ds}+ x_H \D\int_0^\infty{\beta_v(s)\bar{i}_v(s)ds} +dx_H  \vspace{1.0mm},\\
0&\hspace{-150mm}=y_v(t,0) \left( \lambda +\mu + \D\int_0^\infty\D\int_0^\infty{\beta_H(\tau,s)\bar{i}_H(\tau,s)dsd\tau}-(\lambda+\mu) S_v^+ S_H^+ \D\int_0^\infty\D\int_0^\infty{\beta_H(\tau,s)\beta_v(s)\pi_H(\tau,s)\pi_v(s)e^{-\lambda (\tau+s)}dsd\tau}  \right) \\
\qquad\qquad\qquad\qquad\qquad\qquad\qquad\qquad - x_H(t)(\lambda+\mu)S_v^+\D\int_0^\infty \D\int_0^\infty  {\beta_H(\tau,s)\beta_v(s)\bar{i}_v(s)\pi_H(\tau,s)e^{-\lambda \tau} dsd\tau} \vspace{1.0mm}.\\
\end{cases}}
\end{equation}
\vskip-0.09in
\noindent\
Since $x_H(t)\neq 0$ and $y_v(0,t)\neq 0$, we set  the determinant below equal to  zero
\begin{align*} 
\left|\footnotesize{
\begin{array}{ll}
\dfrac{\lambda}{(\lambda+\mu)} + \D\int_0^\infty{\beta_v(s)\bar{i}_v(s)ds} +d
                 &\dfrac{S_H^+}{(\lambda+\mu)}\D\int_0^\infty{\beta_v(s) e^{-\lambda s}\pi_v(s)ds}\\
\\
-S_v^+\D\int_0^\infty \D\int_0^\infty  {\beta_H(\tau,s)\beta_v(s)\bar{i}_v(s)\pi_H(\tau,s)e^{-\lambda \tau} dsd\tau}& 
1 + \tilde{T}- S_v^+ S_H^+ \D\int_0^\infty\D\int_0^\infty{\beta_H(\tau,s)\beta_v(s)\pi_H(\tau,s)\pi_v(s)e^{-\lambda (\tau+s)}dsd\tau}\\
\end{array}}
\right|,
\end{align*}
where $\tilde{T}=\dfrac{T}{(\lambda +\mu)}$, with $T=\D\int_0^\infty\D\int_0^\infty{\beta_H(\tau,s)\bar{i}_H(\tau,s)dsd\tau}.$ Notice that $T \geq 0.$ 
We obtain the following characteristic equation:
\footnotesize{\begin{align}\label{characteristic_equ}
\dfrac{ \lambda+d+ \D\int_0^\infty{\beta_v(s)\bar{i}_v(s)ds} }{\lambda+d}\dfrac{\lambda+\mu+T} {\lambda+\mu}&=S_v^+ S_H^+ \D\int_0^\infty\D\int_0^\infty{\beta_H(\tau,s)\beta_v(s)\pi_H(\tau,s)\pi_v(s)e^{-\lambda (\tau+s)}dsd\tau}  \nonumber\\
\vspace{-.5cm}& +\dfrac{S_v^+ S_H^+ \bar{i}_v(0)}{\lambda+d}\left( \D\int_0^\infty{\beta_v(s)\pi_v(s)ds} \D\int_0^\infty\D\int_0^\infty{\beta_H(\tau,s)\beta_v(s)\pi_H(\tau,s)\pi_v(s)e^{-\lambda (\tau+s)}dsd\tau}\right.  \nonumber\\
\qquad\qquad\qquad\qquad\qquad\qquad\qquad&\left.-\D\int_0^\infty{\beta_v(s)e^{-\lambda s}\pi_v(s)ds} \D\int_0^\infty\D\int_0^\infty{\beta_H(\tau,s)\beta_v(s)\pi_H(\tau,s)\pi_v(s)e^{-\lambda \tau}dsd\tau} \right)  \nonumber\\  
\end{align}}

\noindent \normalsize Now assuming $V(s,V_0^{h \rightarrow v})=V_0^{h \rightarrow v},$ we have $\beta_v(s)=\tilde{\beta}_v$ for some constant $\tilde{\beta}_v=\dfrac{d_0 V_0^{h \rightarrow v}}{d_1+V_0^{h \rightarrow v}}>0.$
Then after cancelling some terms and rearranging \eqref{characteristic_equ}, we obtain the following characteristic equation:
\footnotesize{\begin{align}
\label{characteristic_equ_scl}
\frac{\left( \lambda+d+ \tilde{\beta}_vI_v^+ \right)}{(\lambda+d)}=\frac{\tilde{\beta}_v S_V^+S_H^+ \int_0^\infty{\beta_H(\tau)e^{-\lambda \tau}\pi_H(\tau)d\tau}}{\lambda+\mu+T}, 
\end{align}}
\noindent \normalsize where $\beta_H(\tau)=\beta_H(\tau,V_0^{h \rightarrow v}),\pi_H(\tau)=\pi_H(\tau,V_0^{h \rightarrow v}).$
Now by the way of contradiction, suppose that the characteristic equation \eqref{characteristic_equ_scl} can have a solution $\lambda$ with positive real part. Let $\lambda=a+bi$ and assume $a\geq 0.$ Taking the absolute value of both side of the equality above, we get
\footnotesize{\begin{align*}
\left| \frac{ \lambda+d+\tilde{\beta}_v I_v^+}{\lambda+d}\right|=\frac{\sqrt{(a+d+\tilde{\beta}_vI_v^+)^2+b^2}}{\sqrt{(a+d)^2+b^2}} >1.
\end{align*}}
\noindent \normalsize Notice that
\footnotesize{\begin{align*}
\begin{array}{l} 
\left|\D\frac{\tilde{\beta}_v S_v^+S_H^+ \int_0^\infty{\beta_H(\tau)e^{-\lambda \tau}\pi_H(\tau)d\tau}}{\lambda+\mu+T} \right| \leq \D\frac{\tilde{\beta}_v S_v^+S_H^+ \int_0^\infty{\beta_H(\tau)e^{-a \tau}\pi_H(\tau)d\tau}}{\sqrt{(a+\mu+T)^2+b^2}} \leq \D\frac{\tilde{\beta}_v S_v^+S_H^+ \int_0^\infty{\beta_H(\tau)e^{-a \tau}\pi_H(\tau)d\tau}}{\mu}\\ 
\qquad\qquad\qquad\qquad\qquad\qquad\quad\leq \tilde{\beta}_v \D\frac{S_V^+}{\mu}S_H^+ \int_0^\infty{\beta_H(\tau)\pi_H(\tau)d\tau}=1\\ 
\end{array}
\end{align*}}
\noindent \normalsize For $\lambda$ with nonnegative real part, the LHS of the inequality remains strictly greater than one, while the RHS is strictly smaller than one. Thus, such $\lambda$'s cannot satisfy the characteristic equation (\ref{characteristic_equ_scl}). Hence the endemic equilibrium is locally asymptotically stable whenever it exists.
\end{proof}
\noindent However the stability of $\mathcal E^+$ when $\mathcal R_0>1,$ for general case, is unknown. The interesting question is: \emph{Is it possible that the heterogeneity among the vector infectivity destabilizes the endemic equilibrium, leading to oscillatory dynamics via \emph{Hopf bifurcation}?} From dynamical systems view, it is not uncommon that while structured PDE models can present oscillatory dynamics, but not  a special case of it; for instance the ODE version of the system, where the model parameters are constant \cite{browne2016immune}. The characteristic equation \eqref{characteristic_equ} is too complicated for analysis to infer existence of a Hopf bifurcation, but future work will explore the possibility.

\noindent \normalsize In the presence of a disease, one also would like to understand under what conditions the disease will remain endemic for large time. We say the disease is \textit{uniformly weakly endemic} if there exists some $\tilde{\varepsilon}>0$ independent of the initial conditions such that 
\begin{align*}
\limsup\limits_{t \rightarrow \infty} I(t)>\tilde{\varepsilon}, \text{ whenever } I(0)>0,
\end{align*}
for all solutions of the model. However the disease is \textit{uniformly strongly endemic} if there exists some $\tilde{\varepsilon}>0$ independent of the initial conditions such that 
\begin{align*}
\liminf\limits_{t \rightarrow \infty} I(t)>\tilde{\varepsilon}, \text{ whenever } I(0)>0,
\end{align*}
for all solutions of the model. In the following results, we identify the conditions that result in the prevalence being bounded away from zero. 
\begin{proposition}Assume that within-vector viral load is constant through vector infection period; i.e. $V(s,V_0^{h \rightarrow v})=V_0^{h \rightarrow v}.$ Then if $\mathcal R_0>1,$ then disease is uniformly weakly endemic.
\end{proposition}
\begin{proof}By the way of contradiction, assume that there exists a solution $I_H(t),$ with $I_H(0)>0,$ such that 
$\lim\limits_{t\rightarrow \infty} I_H(t)=0.$
Let $\varepsilon_1>0$ be given. Then $\exists t_0>0: I_H(t) \varepsilon_1, \ \forall t \geq t_0.$  Consequently, the semigroup properties of a solution \eqref{semi} imply that without loss of generality we can assume the above inequality valid for all $t\geq 0.$\\
Next note that
\begin{equation}
\int_0^\infty \int_0^\infty \beta_H (\tau,s)i_H(t,\tau,s)d\tau ds \leq K,
\end{equation}
for some positive real number $K,$ since $\beta_H(\tau)\leq a_0 (:=\hat{\beta}_H)$ and $\limsup\limits_{t\rightarrow \infty} N_H(t) \leq \dfrac{\Lambda}{d}.$\\
Then by the third equation of the system \eqref{SIRhost_n},
\begin{equation}
R^{\prime}_H \leq \hat{\gamma}\varepsilon_1-d R_H.
\end{equation}
Then $\limsup\limits_{t\rightarrow \infty} R_H(t) \leq \dfrac{\hat{\gamma}\varepsilon_1}{d}.$ Hence by the inequality above, we have $R_H(t)\leq \dfrac{\hat{\gamma}\varepsilon_1}{d}+\delta_0,$ for given $\delta_0>0$ and $\forall \ t\geq 0,$ by semigroup property. By similar argument above, we also obtain $\limsup\limits_{t\rightarrow \infty} I_v(t)\leq S^0_v\dfrac{\hat{\beta}_H\varepsilon_1}{\mu}.$ Then $I_H(t), R_H(t),I_v(t) \leq \varepsilon_2,$ where $\varepsilon_2=\dfrac{\hat{\gamma}\varepsilon_1}{d}+\delta_0+S^0_v\dfrac{\hat{\beta}_H\varepsilon_1}{\mu}+\varepsilon_1.$  Next by the first and fourth equation in the system \eqref{SIRhost_n}-\eqref{SIvector_n}, we obtain
\begin{align*}
S^{\prime}_H & = \Lambda -S_H(t)\int_0^{\infty} {\beta_v(s)i_v(t,s)ds}-d S_H(t)  \geq  \Lambda-S_H(t)\hat{\beta}_v \varepsilon_2-dS_H(t).\\
    S^{\prime}_v&=\eta -  S_v(t) \int_0^{\infty}\int_0^{\infty}\beta_H(\tau,s)i_H(t,\tau,s) ds d\tau-\mu S_v(t)  \geq \eta -S_v(t)(\hat{\beta}_H\varepsilon_1+\mu)
\end{align*}
Then 
\begin{align*}
\liminf \limits_{t} S_H(t) \geq \dfrac{\Lambda}{\hat{\beta}_v \varepsilon_2+d}:=S_H^0(\varepsilon_2), \ 
\liminf \limits_{t} S_v(t) \geq \dfrac{\eta}{\hat{\beta}_H \varepsilon_1+\mu}=S_v^0(\varepsilon_1).
\end{align*} 
Since the functions defined above are continuous and $\lim\limits_{\varepsilon_j\rightarrow 0}S^0_i(\varepsilon_j)=S^0_i$, $i \in \{v,H\}, j \in \{1,2\}$, it follows that for given $\varepsilon_3>0, \exists t_1: S_i(t)\geq S^0_i- \varepsilon_3$ for $i \in \{v,H\},$ and $\forall t\geq t_1.$ Again by semigroup property, w.l.o.g. the inequality above is valid for all $t>0.$ 
Now assuming $V(s,V_0^{h \rightarrow v})=V_0^{h \rightarrow v},$ we have $\beta_v(s)=\tilde{\beta}_v$ for some constant $\tilde{\beta}_v=d_0 V_0^{h \rightarrow v}/ (d_1+V_0^{h \rightarrow v})>0.$

\noindent Next note that by the equation \eqref{Iv}, we have
\begin{align}\label{comp_prin}
I^{\prime}_v(t)&\geq (S_v^0-\varepsilon_3)(S_H^0-\varepsilon_3)\tilde{\beta}_v I_v(t)\int_0^\infty{\beta_H(\tau)\pi_H(\tau)d\tau}-\mu I_v(t),
\end{align}
where $\beta_H(\tau)=\beta_H(\tau,V_0^{h \rightarrow v}),\pi_H(\tau)=\pi_H(\tau,V_0^{h \rightarrow v}).$
Thus we can write \eqref{comp_prin} in the following form:
\begin{align}
I^{\prime}_v(t)&\geq \mu I_v(t)\left( \mathcal R_0(\varepsilon_3)-1\right), \label{Ivcomp}
\end{align}
where $\mathcal R_0(\varepsilon_3)\searrow \mathcal R_0 $ as $\varepsilon_3\searrow 0$.  Therefore, since $\mathcal R_0>1$, for sufficiently small $\epsilon_3$, and by comparison principle applied to \eqref{Ivcomp},
$I_H(t), I_v(t) $ goes to infinity, as $t \rightarrow \infty.$ This is a contradiction to boundedness of solutions.
\end{proof}

\noindent Here we conjecture that, for general case,  $\mathcal R_0>1$ implies \emph{uniform persistence} of disease, and preserve it as future work. We turn our attention to crucial extensions of this modeling framework, and epidemiological implications.\\
\begin{SCfigure}\label{Within_vectorfitting}
 \centering
 \caption{The within-mosquito virus dynamics with fitted parameters at days post vector infection $s \in [0, 32 ]$ and the pathogen load data in midgut of mosquito, presented by blue dots. The fitted parameter values are: $r_v= 0.3258,\ K_v= 1.2303\times 10^3, \ U_v= 0.9933.$} 
\includegraphics[width=0.41\textwidth]{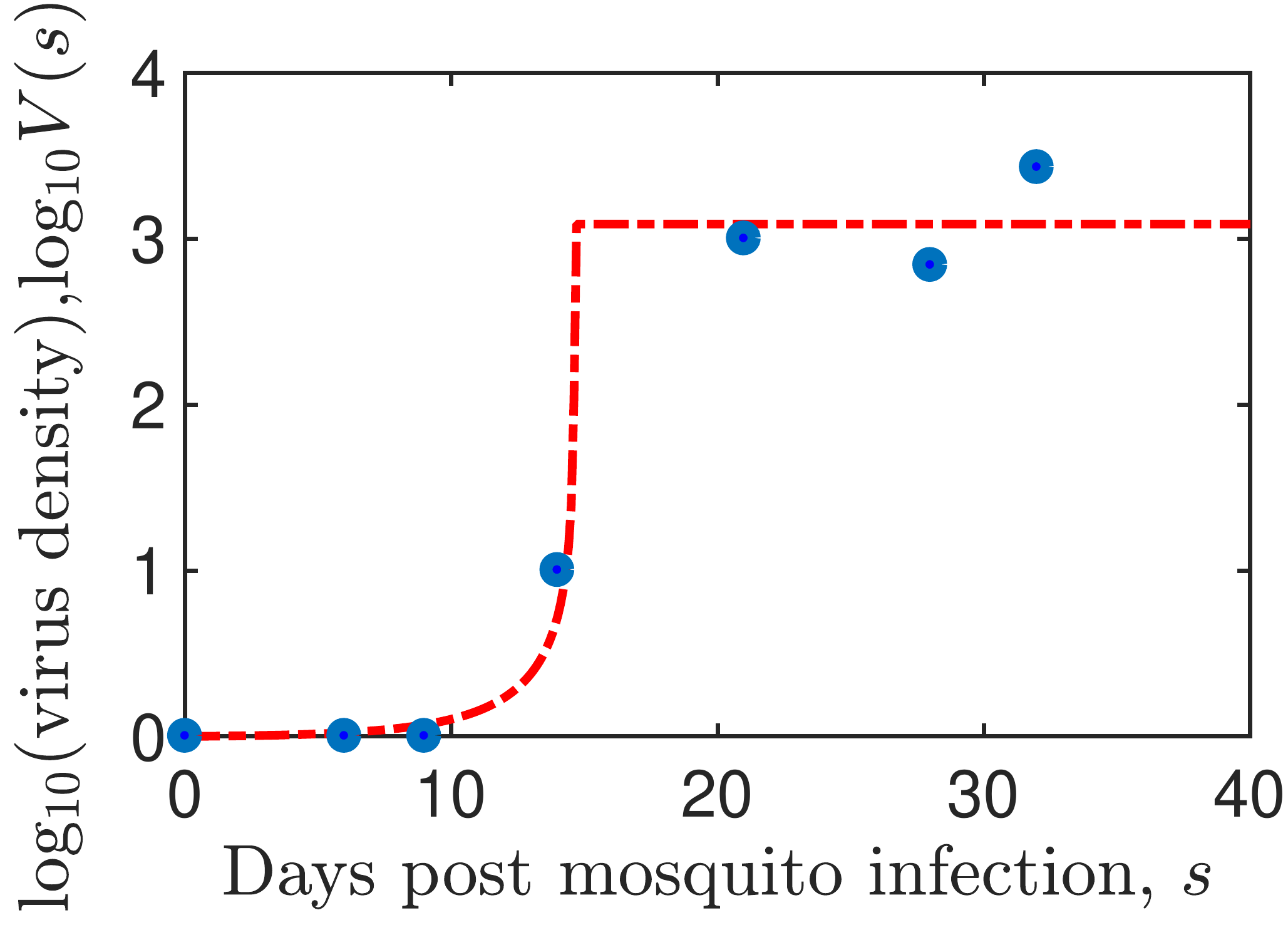}
\end{SCfigure} 

\begin{figure}[t!]\label{Within_vector_pathogen_vs_R0}
\begin{center}
(a)\includegraphics[width=0.44\textwidth]{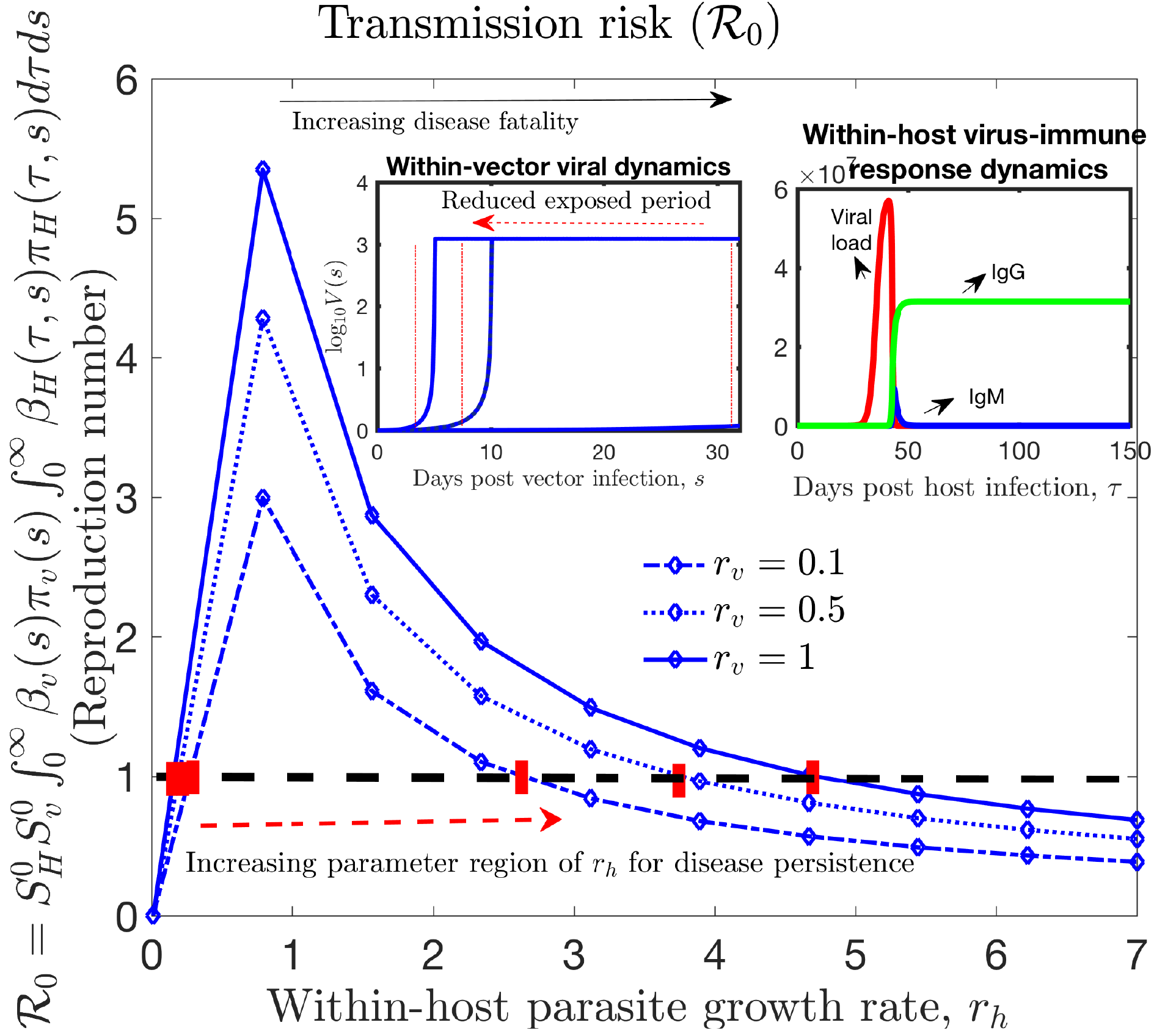} 
\mbox{\hspace{0.01\textwidth}}
b)\includegraphics[width=0.45\textwidth]{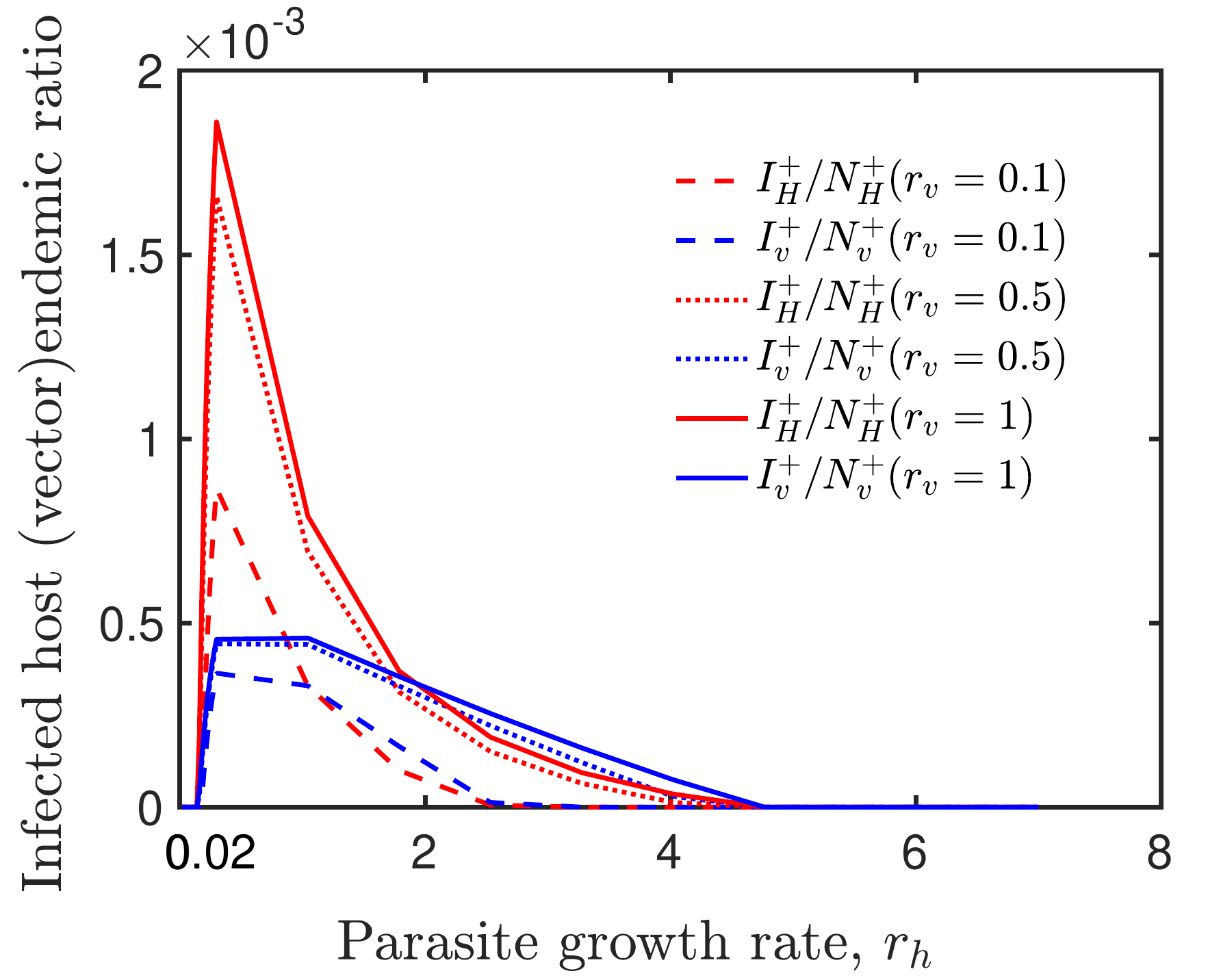}
\caption{ a)Basic reproduction number, $\mathcal R_0,$ versus within-host parasite growth rate, $r_h,$ with respect to distinct values of within-vector parasite growth rate, $r_v$. An increase in within-vector viral growth rate causes an increase in the parameter range of within-host viral growth rate, leading disease persistence (increasing the initial transmission risk $\mathcal R_0$ above one).  The inserted figures show the corresponding within-vector (left) and within-host (right) dynamics for the given parameter value of $r_v.$  b) The corresponding endemic ratio for steady state disease abundance $ \bar{\mathcal I}_H$ (host) and $\bar{\mathcal I}_v$ (vector population) w.r.t. varying values of  $r_v$ and $r_h.$}
\end{center}
\end{figure}

 
\begin{figure}[t!]\label{fig_R0}
\begin{center}
\includegraphics[width=0.9\textwidth]{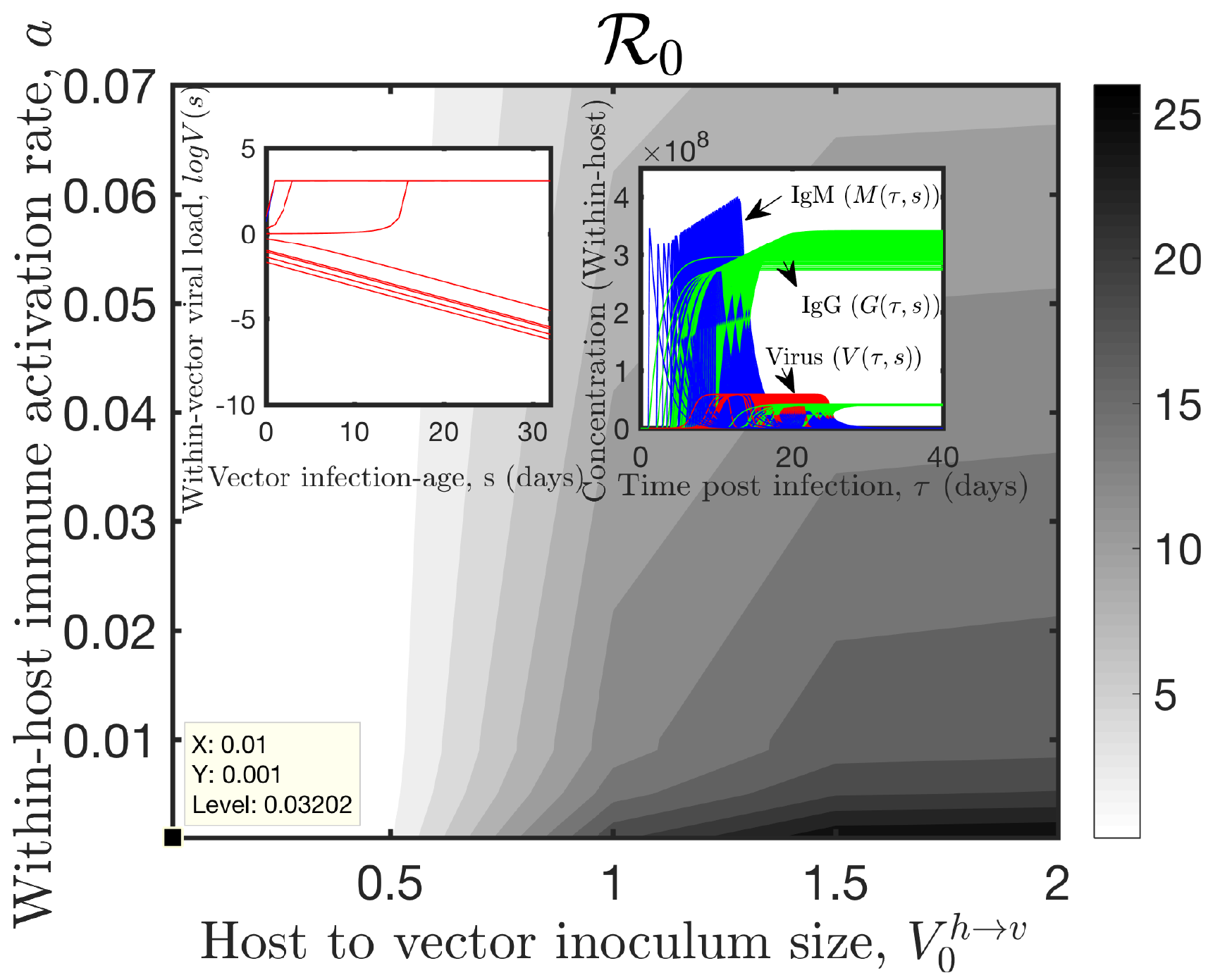}
\caption{\emph{Changing transmission risk across the parameter values.} \textbf{Main figure:} $\mathcal R_0$ versus varying host to vector inoculum size $V_0^{h \rightarrow v}$  and within-host IgM immune activation rate $a.$ \textbf{Left inserted subfigure:} It displays within-vector viral dynamics with respect to host to vector inoculum size $V_0^{h \rightarrow v}\in \{0.01, 0.05, 0.1, 0.2, 0.4, 0.5, 1, 1.5, 2 \}.$ Bistability is due to Allee effect.  \textbf{Right inserted subfigure:} It depicts time-since-infection dependent immune variables' distributions for given initial conditions $P_0^{v \rightarrow h}$ (resulted from within-vector viral kinetics,see left inserted subfigure) and for varying within-host immune activation parameter $a$ values. Recall that vector to host inoculum size (initial condition for within-host viral load) is: $P_0^{v \rightarrow h}=\hat{c} V(s,V_0^{h \rightarrow v}).$ Feedback across within-vector, within-host and between host scale disease transmission provides distinct regimes of initial transmission risk as shown by the contour plot of $\mathcal R_0$ in the main figure.} 
\end{center}
\end{figure} 
 

\section{The signatures of within-vector viral kinetics on disease dynamics}\label{implications}

 An interesting question is:  \textit{Is it possible that infectiousness of mosquitoes can be a good predictor of disease outbreaks?}  In a recent study, Churcher et al.\cite{churcher2017probability} found that the amount of parasites in a mosquito's salivary glands not only is a good indicator for how much the mosquito bite can be infectious, and also how faster infection would develop within-host upon receiving the bite. Therefore it is crucial to understand \emph{how the within-vector viral kinetics can be scaled up to the disease dynamics among host population for prediction and disease intervention.}
 
\noindent In previous section, we show that the basic reproduction number, $\mathcal R_0,$ is a threshold quantity, providing whether a disease can persist or  eventually  die out.  To numerically calculate the initial transmission risk $\mathcal R_0$, we first compute the probability function  $\pi_H(\tau,s),$ depending on the within-host model  \eqref{PB} variables, with initial condition, vector to host inoculum size $P_0^{v\rightarrow h}=h(V(s_i,\bfa p_v)),$ where $s_i$  is the vector-infection-age, and $\bfa p_v$ represent the within-vector model parameters.
Notice that the epidemiological parameters $\beta_H(\tau,s), \nu_H(\tau,s), \gamma_H(\tau,s)$ are functions of within-host immunological variables $P(\tau,s), M(\tau,s), G(\tau,s),$ and the vector to host inoculum size, $P_0^{v\rightarrow h},$ governed by within-vector variable $V(s,\bfa p_v)$. Let $V(s,\bfa p_v)$ be the solution of the system \eqref{within_vector} with initial pathogen concentration $V_0^{h \rightarrow v}$. Here, we consider the host to vector inoculum size, $V_0^{h\rightarrow v},$ to be constant, representing the mean.  Then by implementing trapezoidal rule multiple times with chosen fixed time step size $\Delta \tau =0.0005,$ we estimated $\mathcal R_0$ and the steady state disease abundance $ \bar{\mathcal I}_H$ (host) and $\bar{\mathcal I}_v$ (vector population) (see Fig.3).

\noindent For numerical simulations, we obtain the value of within-host virus-immune response model and epidemiological parameters from the literature, presented in Tables  \ref{table:fitparamPB}, \ref{table:fixparam2}, and \ref{table:fitparamepi}, respectively. To estimate the within-vector parameters, we extracted within-mosquito WNV viral data, given in \cite{fortuna2015experimental}, by using MATLAB code \emph{grabit.m}, and numerically fit these data by using the \textit{least square error}.  Fig.2 displays the fitted model solution and the within-mosquito viral data (blue dots) given in \cite{fortuna2015experimental}. The fitted parameter values are: $r_v= 0.3258,\ K_v= 1.2303\times 10^3, \ U_v= 0.9933.$ 

\noindent Prior field studies suggest that environmental factors can manipulate the mosquito's competence \cite{tabachnick2013nature,anderson2010effects}. For example, it has been shown that as temperature increases,  virus replication increases in a mosquito's tissues; therefore increasing viral replication within-mosquito and viral transmission to host. Some of the observed effects of temperature, it is suggested, are due to increased viral replication at higher temperatures that often results in a shortening of the EIP (the time it takes for a mosquito to become infectious once it has taken a viremic blood meal).  
In addition, it is revealed that the host to vector inoculum size and the length of the exposed period influence the effect of WNV transmission by  Cx. nigripalpus \cite{anderson2010effects}.  
However the impact of these factors on disease dynamics is still unclear.


\noindent In Fig.3, we vary the value of within-vector viral growth rate as $r_v=0.1$ (orange line), $r_v=0.5$ (green line), $r_v=1$ (blue line), and plot the corresponding values of basic reproduction number for distinct values of within-host viral growth rate $r_h \in [0.01, 7]$, and ask: \emph{How does the effect of external factors such as temperature can be scaled up to host population level disease transmission?}
The inserted subfigures in Fig.3 displays the corresponding within-vector viral dynamics with respect to varying values of the vector parameter $r_v$(right), and the within-host viral-immune response antibody dynamics for $r_h=7.$ and $P_0^{v \rightarrow h}=0.01.$  We observe that increasing value of $r_v$ shortens virus incubation period for vectors, mimicking field studies, mentioned above \cite{tabachnick2013nature}. Our numerical results suggest that in return, at host population scale, these mechanisms may lead to significant increase in initial transmission risks, $\mathcal R_0.$ The biological insight is that increasing viral replication rates inside the mosquito decreases the time needed for a blood-fed mosquito to be able to pass on the virus to another host \cite{dohm2002effect, soverow2009infectious, reisen2014effects}. For example, in Fig.3, an increase in within-vector parasite growth rate $r_v$ from $0.5$ to $1,$ shortens the incubation period for $4$ days (see inserted left figure), and increases the epidemic (disease persistence) parameter range of $r_h$ ($\mathcal R_0>1$) from $[0.15 \ 3.9]$ to $[0.14 \ 4.68].$  Therefore these results demonstrate that  in an environment where disease may not persist due low susceptibility of host population (lower $r_h$), an increase in within-vector viral replication rate $r_v$ via external factors may lead prolonged epidemic (disease persistence). For instance, in an environment, where a resident mosquito population has a mean value of $r_v=0.5,$ the disease only persists when within-host viral growth rate $r_h$ is in the parameter range $ [0.15 \ 3.9].$ However when within-vector parasite growth rate is increased to $r_v=1$, this range increases to $[0.14 \ 4.68],$ implying that the infectivity of vector population can make host population more susceptible to epidemics. Therefore  within-vector viral kinetics (whether it is manipulated by external factors or not) can change the fate of the disease outcomes, and might be a good predictor for disease outbreaks.

In Fig.4, we also assess how the initial transmission risk, $\mathcal R_0,$ changes across the distinct values of  $V_0^{h \rightarrow v}$ and the within-host immune response parameter for $a.$  The numerical results suggest that when the immune activation parameter $a$ is small; i.e. when host population does not have strong immunity or protection against infection, the disease dynamics is very sensitive to host population infectivity. This implies that the impact of host to vector inoculum size $V_0^{h \rightarrow v}$ on the disease dynamics among the host population is more magnified in a host population with low immune profile. When there is no sufficient host immune response, larger $V_0^{h \rightarrow v}$ increases the probability of the parasite transmission from host to vector, in return increases the transmission risk among the host population significantly. These findings could have significant implications for public health, magnifying importance of control strategies such as drug treatment or vaccination, which can be utilized to slow down the viral production within-host scale.

\noindent Another crucial motivation for considering within-vector viral dynamics explicitly in a tractable system is to assess the impact of vector to host inoculum size on the efficacy of control strategies. For example, the recent evidence suggests that the vaccine was less effective when mice or humans were bitten by mosquitoes carrying a greater number of parasites. Assuming that vaccination mainly works by increasing the within-host immune response activation rates $a,$ (or $b$), our results (mentioned above) mimic the observations from the field studies as follows: larger host to vector inoculum size $V_0^{h \rightarrow v}$ results in shorter vector exposed period (see left inserted subfigure in Fig.4), and subsequently generates infectious vector distribution with larger inoculum size $P_0^{v \rightarrow h}.$ Among host population with low immune state, this resulting increase in  $P_0^{v \rightarrow h}$ causes larger vector to host virus transmission, ultimately leading an increase in the number of secondary cases ($\mathcal R_0$). The biological insight behind of these findings from field studies is that because the vaccine can only kill a certain proportion of the parasites, it is overwhelmed when the parasite population is too large, suggesting that "it will become epidemiologically important to know how infected a mosquito is for disease elimination". 

\noindent One drawback of our model (\ref{SIRhost_n}-\ref{SIvector_n}) is that the host to vector inoculum size $V_0^{h \rightarrow v}$ is chosen to be constant, representing the mean. Indeed it should also depend on the within-host viral load $P(\tau,s).$ Next we argue the motivation and challenges in overcoming this strain, and develop a feasible way to vary the host to vector  inoculum size $V_0^{h \rightarrow v},$ depending on host infectiousness. Recall that in our model, vector to host inoculum size depends on within-vector pathogen dynamics; i.e. $P^{v \rightarrow h}_0=h(V(s, V_0^{v\rightarrow h}))$, but host to vector inoculum size is assumed to be constant, $V_0^{v\rightarrow h}=P_0.$

\begin{SCfigure}\label{Infection_cycle}
\centering
 \includegraphics[width=0.7\textwidth]{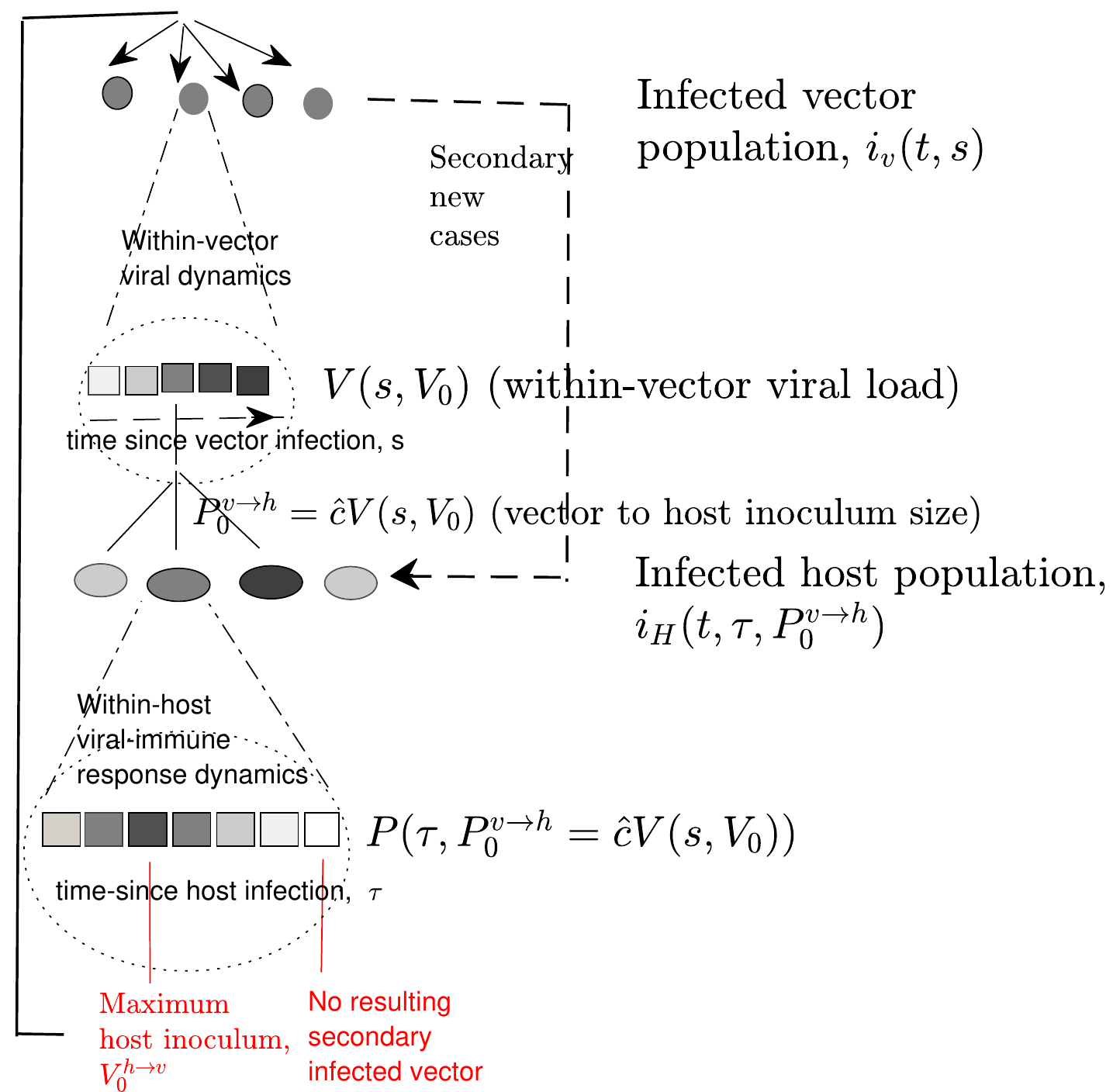}
 \caption{\emph{``Infinite-dimensional'' feedback.} The chain of across-scale interactions effectively introduces feedback from both scales. } 
\end{SCfigure} 

\begin{figure}[t!]\label{fig_R01upd2}
\begin{center}
(a)\includegraphics[width=0.35\textwidth]{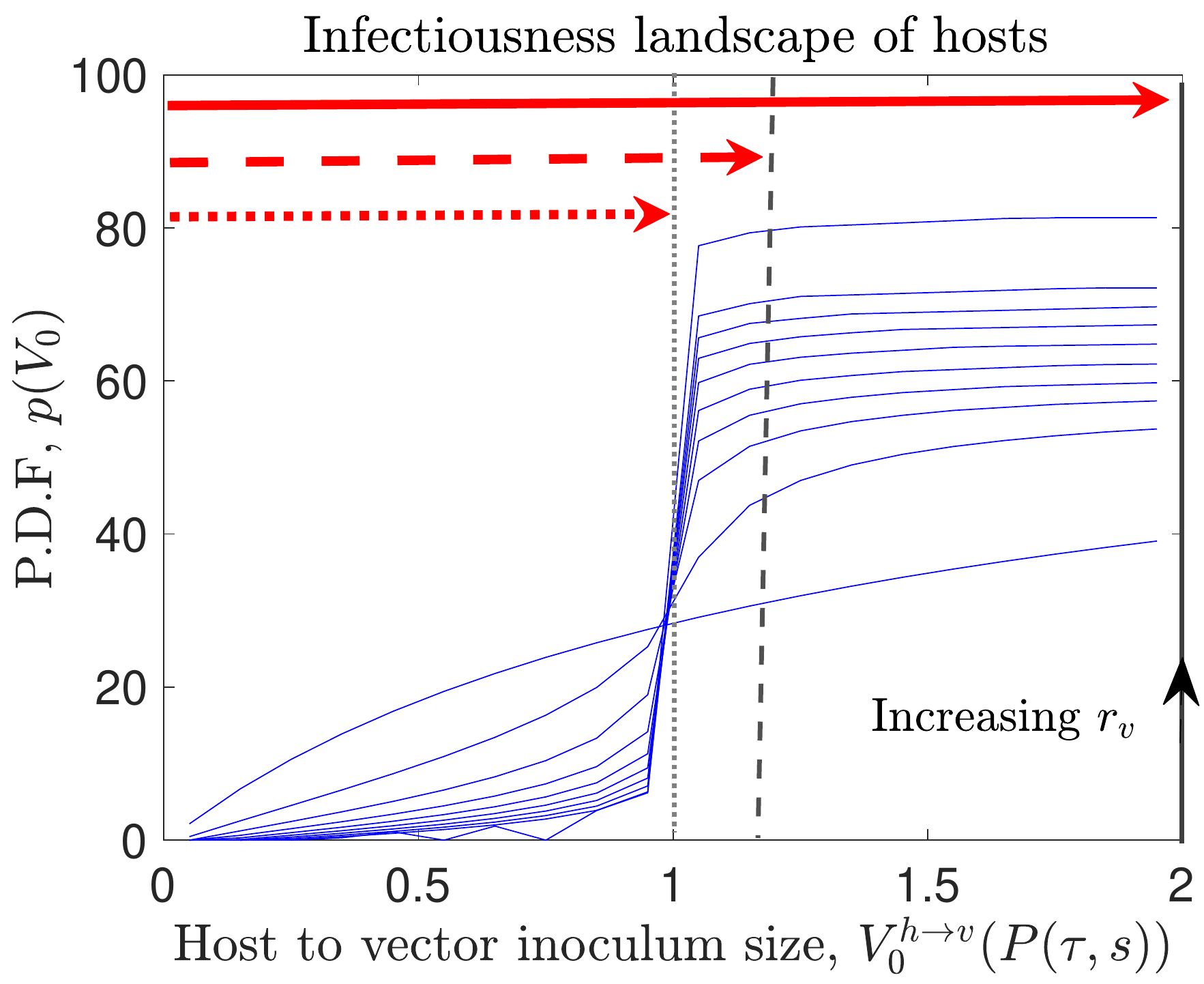}
\mbox{\hspace{0.01\textwidth}}
(b)\includegraphics[width=0.4\textwidth]{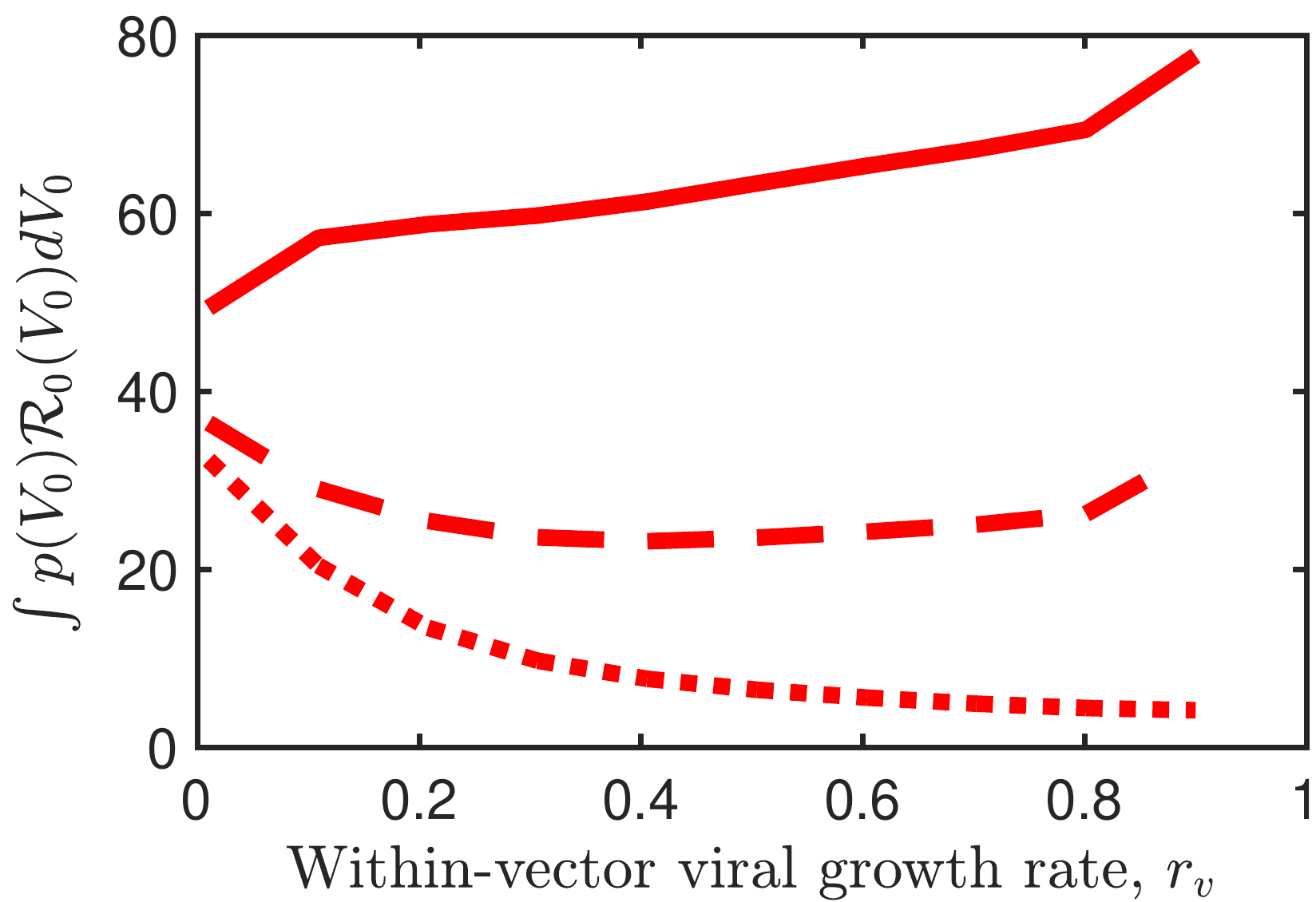}
\caption{a)The distribution of host to vector inoculum size $V_0^{h \rightarrow v}$ with respect to varying within-vector viral growth rate $r_v.$ b) $\mathcal R_0$ versus the within-vector viral growth rate $r_v$ with highly infectious (solid line, $V_0^{h \rightarrow v} \in [0,2]$), medium infectious  (dashed line, $V_0^{h \rightarrow v} \in [0,1.2]$), and less infectious host population (dotted line, $V_0^{h \rightarrow v} \in [0,1]$), respectively.
} 
\end{center}
\end{figure} 

\subsection{Incorporating two-way feedback between epidemiological and immunological scales}\label{two_way_feedback}


\noindent Mosquito vectors are often exposed to hosts that individually vary in pathogen loads, which can result in variation in the proportion of the mosquito population that becomes infectious. Although recent immuno-epidemiological models \cite{gulbudak2017vector, tuncer2016structural} have been successful in measuring host to vector transmission as a function of within-host viral load, in these models the varying within-host pathogen loads only implicitly affect the number of secondary infectious cases, where impact of host to vector inoculum size cannot be studied explicitly.
Here we introduce a feasible way to incorporate host to vector inoculum size $V_0^{h \rightarrow v},$ by defining a distribution. It affects within-vector viral kinetics, subsequently vector-to-host disease transmission chain. Finally, we numerically and analytically explored how the infectivity of host population (measured by $V_0^{h \rightarrow v}$) impacts arbovirus disease dynamics.

\noindent Notice that host to vector inoculum size, $V_0^{h \rightarrow v}$, which is a function of within-host viral load $P(\tau,s)$,  affect the within-vector dynamics, described by $V^\prime(s)=g(V(s),V_0^{h \rightarrow v}).$ It is in return affect the distribution of the vector to host inoculum size $P_0^{v\rightarrow h}=h(V(s,V_0^{h \rightarrow v})),$ governing the within-host dynamics and ultimately between-host disease dynamics.  This chain of across-scale interactions effectively introduces feedback from both scales (See Fig.5).  Yet, due to the significant mathematical complexity of a two-way ``infinite-dimensional'' feedback, closing the loop is not feasible. Next we describe a  feasible way to improve this limitation. In contrast to previous attempts of incorporating feedback between scales \cite{gandolfi2015epidemic}, our approach is amenable to analysis and biologically relevant for vector-borne diseases.

\noindent To define the distribution, we assume that the probability distribution of $V_0^{h \rightarrow v}$ depend on two factors: (i) the amount of pathogen that a susceptible vector might get upon biting an infected host, depending on \emph{the host infectiousness} $\beta_H(P(\tau,s), P_0^{v\rightarrow h})$ and (ii) the \emph{infected host density} $\bar{i}_H(\tau, s, P_0^{v\rightarrow h}).$ Therefore, we consider a distributed $V_0^{h \rightarrow v},$ depending on within-host dynamics, given by \eqref{PB}, and the corresponding steady state infected host density, $\bar{i}_H(\tau, s,P_0^{v\rightarrow h}),$ given by \eqref{vectorhostequ.}, as follows:

\begin{equation}\label{PDF}
p(V_0^{h \rightarrow v})=\dfrac{\int_0^{\infty}\int_0^{\infty}{\beta_h(P(\tau, P_0^{v\rightarrow h}=h(V(s, V_0))))\bar{i}_H(\tau, s,P_0^{v\rightarrow h}=h(V(s, V_0)))d\tau}ds}{\int_{V^0_{lower}}^{V^0_{upper}}\int_0^{\infty} \int_0^{\infty}{\beta_h(P(\tau, P_0^{v\rightarrow h}=h(V(s, V_0))))\bar{i}_H(\tau, s, P_0^{v\rightarrow h}=h(V(s, V_0)))ds d\tau}d V_0}.
\end{equation} 

\noindent Then the system has the basic reproduction number as follows:
\begin{equation*}
\int {p(V_0^{h \rightarrow v})\mathcal R_0(V_0^{h \rightarrow v}, a) dV_0^{h \rightarrow v}}.
\end{equation*}


\noindent The Fig. 6(a) displays the distribution of host to vector inoculum size with respect to varying within-vector viral growth rate $r_v$. The Fig. 6(b) displays how $\mathcal R_0$ changes w.r.t. varying $r_v,$ given how infectious the host population is, which is defined by the range of $V_0^{h \rightarrow v}$.  Numerical results suggest that in a highly infectious host population (larger inoculum  size range: $V_0^{h \rightarrow v} \in [0.01\ 2] $), increasing within-vector viral growth rate $r_v$ increases initial transmission risk, $\mathcal R_0.$ However, an increase in $r_v$ among a less infectious host population (inoculum  size $V_0^{h \rightarrow v} \in [0.01\ 1] $) decreases $\mathcal R_0,$ suggesting that when the host population is less infectious due to resulting distribution $p(V_0^{h \rightarrow v}),$ vectors as intermediate carries only dilute the effect of disease transmission due to Allee effect.

\noindent These findings could have significant implications for disease control.  Our numerical results highlight that:  (i) a significant reduction in host to vector inoculum size, which can be accomplish utilizing a drug treatment or vaccination (to hamper virus transmission from host to vector) can reduce $\mathcal R_0,$ significantly, in which case the disease can be ultimately eradicated (see Fig.6), (ii) when host population is very infectious, vector infectivity magnifies the disease outcomes.

\noindent A more rigorous approach in assessing the impact of vector, or host parameters on disease dynamics requires investigating the sensitivity of within-vector viral kinetics, within-host immune response and epidemic parameters to $\mathcal R_0,$ and $  \mathcal E^+$ which we reserve as a future work.

\section{Conclusion}\label{conc}

Within-vector viral dynamics can be a driving mechanism in disease dynamics and vector-borne pathogen evolution. Assessing the impact of within-vector viral kinetics on disease dynamics at population scale requires tractable models to measure this impact. In this study, we develop a multi-scale vector-borne disease model, connecting all scales from within-vector viral kinetics to between vector-host disease spread. By doing so, we investigate the impact of within-vector viral kinetics on disease dynamics and, in particular, address:
\begin{itemize}
\item[(i)]  How host and vector infectivity, measured by inoculum size, might affect the landscape of the initial transmission risks, $\mathcal R_0.$
\item[(ii)]  And how within-host immune response combined with within-vector viral kinetics might affect the success of a control strategy such as vaccination, and drug treatment.
\end{itemize}

 There are several reasons for explicitly modeling the heterogeneity in the within-vector dynamics in the novel manner of this work.  First the overall aim is to construct a multi-scale model so variations in parameters associated with within-vector viral kinetics can be extrapolated to the overall epidemic dynamics.  These variations may come from climate or environmental factors, or bio-control strategies such as \emph{Wolbachia}, and the within-vector viral model can be validated directly from experiment.  Tracking dynamics within vectors, as opposed to using average quantities for vector epidemiological parameters (as in ODE models), allows for important biological and mathematical features to be captured.  For instance, the delay between infecting bite and viral growth to infectious levels within the mosquito, known as \emph{ extrinsic incubation period (EIP)}, can be a very sensitive quantity for determining disease spread \cite{tjaden2013extrinsic}. While a delay differential equation (DDE) can be used and is also a special case of the PDE, a DDE will still not capture all of the heterogeneity levels of infectiousness that this model.  Furthermore, our multi-scale framework provides a natural way to have within-vector parameters shape the EIP, and, in turn, the  overall epidemic, as exemplified in Fig. $3$ where within-vector viral growth rate can largely affect $\mathcal R_0$.

In addition, a goal of this work is to model how vector-to-host inoculum size affects the dynamics.  In previous work \cite{gulbudak2017vector}, we showed that this inoculum size, as a parameter, has a very large effect on virulence evolution for vector-borne diseases.  Here by also incorporating vector viral infection kinetics, we can directly model variable inoculum size based on the within-vector dynamics determining both infectiousness and the initial condition in our within-host model which has a large impact on the host infection, as observed in Fig. $4$.

For analytical results, upon defining the basic reproduction number, $\mathcal R_0,$  (depending on host and vector infectious status), we prove that if $\mathcal R_0<1,$ the disease-free equilibrium  $\mathcal E_0$ is locally (via linearization) and globally asymptotically stable (via comparision principle). Otherwise if $\mathcal R_0>1,$ the system has a unique endemic equilibrium, $\mathcal E^\dagger,$ and it is locally asymptotically stable when the vector to host  inoculum size, $P_0^{v \rightarrow h},$ is constant during vector infectious time period. However, for general case, the stability of $\mathcal E^\dagger,$  might not be guaranteed via standard linearization method. We provide a condition that if holds, the system might present a Hopf bifurcation, leading oscillatory dynamics.  Given the constant vector to host inoculum size, we also show that whenever $\mathcal R_0>1,$ the disease is uniformly weakly persistent.

Our numerical results suggest that when immune response is very low among host population, or when vaccination do not provide sufficiently large immunity, the disease transmission between and within-host are very sensitive to vector to host inoculum size $P_0^{v \rightarrow h}$, mimicking field studies.  Indeed, recent field studies suggest that when mosquitoes were very infectious (large vector to host inoculum size), the vaccine was less effective when mice or humans were bitten by mosquitoes carrying a greater number of Malaria parasites, due to "overwhelmed" immune response. Therefore the within-vector viral kinetics, providing  how infectious mosquito population,  can be crucial determining the transmission risk, and it can impact the outcomes of disease control strategies such as vaccination. In addition, we extend this model by incorporating a distribution of host to vector inoculum size $V_0^{h \rightarrow v}.$ Our results suggest that in a highly infectious host population, the disease outbreaks are highly sensitive to vector competence, magnifying the importance of disease control strategies such as drug treatment, and vaccination, which can slow down the viral progression within-hosts. 

In conclusion, in this paper, we develop an immuno-epidemiological model, coupling within-vector viral kinetics, within-host virus-immune response and between host -vector disease transmission.  As one of the crucial applications, we show that the developed multi-scale model can be utilized to assess the impact of disease control strategies, when the infectiousness of host population and vector population vary across scales. We also investigate how environmental factors such as temperature can magnify the role of vectors in disease outcomes. 
Field studies also suggest that nutrition and competition during the larval stage may also influence the transmission capability of arboviruses for the resulting adult females. In addition, environmental factors such as exposure to insecticides in the adult or larval stages has been shown to influence mosquito competence for arboviruses  \cite{muturi2011larval,yadav2005effect,muturi2011larval}. Future work will investigate the role of these mechanisms on disease dynamics. 
In addition in host scale, we only consider adaptive, but not innate immune response, which might affect the disease outcomes.  Future work will include these complexities. The multi-scale modelling framework, introduced here, can be utilized to assess the role of vectors on disease dynamics, given many external factors affecting the vector competence and host immune response. In addition, it can be used for assessing the impact of Wolbachia-based biocontrol strategy, mainly utilized to interfere within-vector viral growth to slow down disease transmission among host population, which will be the future work. In summary,  the modeling work contained in this paper can help to understand the effect of within-vector viral kinetics on arbovirus disease dynamics and help to guide policies on strategies for disease control.



\section{Acknowledgment} This project may have benefited from discussions with Gabriela Blohm (University of Florida, College of Public Health and Health Professions). In addition, the author also thanks two anonymous reviewers for their helpful comments and feedback on the manuscript, and James M. Hyman (Tulane University), Carrie Manore (Los Alamos National Laboratory), and Gerardo Chowell (Georgia State University) for their suggestions and comments during \emph{ New Orleans workshop on Modeling the Spread of Infectious Diseases} at Tulane University. This work was supported by a grant from the Simons Foundation/SFARI($638193$, HG).

\newpage
\setcounter{page}{1}
\bibliographystyle{siamplain.bst}
\bibliography{NSF_References}

\begin{thebibliography}{10}

\bibitem{agusto2019transmission}
F.~B. Agusto, M.~Leite, and M.~E. Orive.
\newblock The transmission dynamics of a within-and between-hosts malaria
  model.
\newblock {\em Ecological Complexity}, 38:31--55, 2019.

\bibitem{anderson2010effects}
S.~L. Anderson, S.~L. Richards, W.~J. Tabachnick, and C.~T. Smartt.
\newblock Effects of west nile virus dose and extrinsic incubation temperature
  on temporal progression of vector competence in culex pipiens
  quinquefasciatus.
\newblock {\em Journal of the American Mosquito Control Association},
  26(1):103, 2010.

\bibitem{browne2016immune}
C.~Browne.
\newblock Immune response in virus model structured by cell infection-age.
\newblock {\em Mathematical Biosciences \& Engineering}, 13(5):887--909, 2016.

\bibitem{cai2017global}
L.-M. Cai, X.-Z. Li, B.~Fang, and S.~Ruan.
\newblock Global properties of vector--host disease models with time delays.
\newblock {\em Journal of mathematical biology}, 74(6):1397--1423, 2017.

\bibitem{cai2013dynamical}
L.-M. Cai, X.-Z. Li, and Z.~Li.
\newblock Dynamical behavior of an epidemic model for a vector-borne disease
  with direct transmission.
\newblock {\em Chaos, Solitons \& Fractals}, 46:54--64, 2013.

\bibitem{churcher2017probability}
T.~S. Churcher, R.~E. Sinden, N.~J. Edwards, I.~D. Poulton, T.~W. Rampling,
  P.~M. Brock, J.~T. Griffin, L.~M. Upton, S.~E. Zakutansky, K.~A. Sala, et~al.
\newblock Probability of transmission of malaria from mosquito to human is
  regulated by mosquito parasite density in na{\"\i}ve and vaccinated hosts.
\newblock {\em PLoS pathogens}, 13(1):e1006108, 2017.

\bibitem{ciupe2017host}
S.~M. Ciupe and J.~M. Heffernan.
\newblock In-host modeling.
\newblock {\em Infectious Disease Modelling}, 2(2):188--202, 2017.

\bibitem{dohm2002effect}
D.~J. Dohm, M.~L. O'Guinn, and M.~J. Turell.
\newblock Effect of environmental temperature on the ability of culex pipiens
  (diptera: Culicidae) to transmit west nile virus.
\newblock {\em Journal of medical entomology}, 39(1):221--225, 2002.

\bibitem{eikenberry2018mathematical}
S.~E. Eikenberry and A.~B. Gumel.
\newblock Mathematical modeling of climate change and malaria transmission
  dynamics: a historical review.
\newblock {\em Journal of mathematical biology}, 77:857--933, 2018.

\bibitem{fan2010impact}
G.~Fan, J.~Liu, P.~Van~den Driessche, J.~Wu, and H.~Zhu.
\newblock The impact of maturation delay of mosquitoes on the transmission of
  west nile virus.
\newblock {\em Mathematical biosciences}, 228(2):119--126, 2010.

\bibitem{forrester2014arboviral}
N.~Forrester, L.~Coffey, and S.~Weaver.
\newblock Arboviral bottlenecks and challenges to maintaining diversity and
  fitness during mosquito transmission.
\newblock {\em Viruses}, 6(10):3991--4004, 2014.

\bibitem{fortuna2015experimental}
C.~Fortuna, M.~E. Remoli, M.~Di~Luca, F.~Severini, L.~Toma, E.~Benedetti,
  P.~Bucci, F.~Montarsi, G.~Minelli, D.~Boccolini, et~al.
\newblock Experimental studies on comparison of the vector competence of four
  italian culex pipiens populations for west nile virus.
\newblock {\em Parasites \& vectors}, 8(1):463, 2015.

\bibitem{franz2015tissue}
A.~Franz, A.~Kantor, A.~Passarelli, and R.~Clem.
\newblock Tissue barriers to arbovirus infection in mosquitoes.
\newblock {\em Viruses}, 7(7):3741--3767, 2015.

\bibitem{fraser2014virulence}
C.~Fraser, K.~Lythgoe, G.~E. Leventhal, G.~Shirreff, T.~D. Hollingsworth,
  S.~Alizon, and S.~Bonhoeffer.
\newblock Virulence and pathogenesis of hiv-1 infection: an evolutionary
  perspective.
\newblock {\em Science}, 343(6177):1243727, 2014.

\bibitem{gandolfi2015epidemic}
A.~Gandolfi, A.~Pugliese, and C.~Sinisgalli.
\newblock Epidemic dynamics and host immune response: a nested approach.
\newblock {\em Journal of mathematical biology}, 70(3):399--435, 2015.

\bibitem{garira2019coupled}
W.~Garira and D.~Mathebula.
\newblock A coupled multiscale model to guide malaria control and elimination.
\newblock {\em Journal of theoretical biology}, 475:34--59, 2019.

\bibitem{gilchrist2002modeling}
M.~A. Gilchrist and A.~Sasaki.
\newblock Modeling host--parasite coevolution: a nested approach based on
  mechanistic models.
\newblock {\em Journal of Theoretical Biology}, 218(3):289--308, 2002.

\bibitem{gulbudak2014modeling}
H.~Gulbudak.
\newblock {\em Modeling culling and vaccination in poultry with application to
  avian influenza}.
\newblock PhD thesis, University of Florida, 2014.

\bibitem{DENVimmunoepi2019}
H.~Gulbudak and C.~Browne.
\newblock Two-strain multi-scale dengue model structured by dynamic host
  antibody level.
\newblock {\em Preprint}, 2019.

\bibitem{gulbudak2017vector}
H.~Gulbudak, V.~L. Cannataro, N.~Tuncer, and M.~Martcheva.
\newblock Vector-borne pathogen and host evolution in a structured
  immuno-epidemiological system.
\newblock {\em Bulletin of mathematical biology}, 79(2):325--355, 2017.

\bibitem{hale1963stability}
J.~K. Hale.
\newblock A stability theorem for functional-differential equations.
\newblock {\em Proceedings of the National Academy of Sciences of the United
  States of America}, 50(5):942, 1963.

\bibitem{handel2015crossing}
A.~Handel and P.~Rohani.
\newblock Crossing the scale from within-host infection dynamics to
  between-host transmission fitness: a discussion of current assumptions and
  knowledge.
\newblock {\em Phil. Trans. R. Soc. B}, 370(1675):20140302, 2015.

\bibitem{honjo2002molecular}
T.~Honjo, K.~Kinoshita, and M.~Muramatsu.
\newblock Molecular mechanism of class switch recombination: linkage with
  somatic hypermutation.
\newblock {\em Annual review of immunology}, 20(1):165--196, 2002.

\bibitem{kitagawa2019mathematical}
K.~Kitagawa, T.~Kuniya, S.~Nakaoka, Y.~Asai, K.~Watashi, and S.~Iwami.
\newblock Mathematical analysis of a transformed ode from a pde multiscale
  model of hepatitis c virus infection.
\newblock {\em Bulletin of mathematical biology}, 81(5):1427--1441, 2019.

\bibitem{lashari2011global}
A.~A. Lashari and G.~Zaman.
\newblock Global dynamics of vector-borne diseases with horizontal transmission
  in host population.
\newblock {\em Computers \& Mathematics with Applications}, 61(4):745--754,
  2011.

\bibitem{martcheva2013unstable}
M.~Martcheva and O.~Prosper.
\newblock Unstable dynamics of vector-borne diseases: Modeling through
  delay-differential equations.
\newblock In {\em Dynamic models of infectious diseases}, pages 43--75.
  Springer, 2013.

\bibitem{martcheva2003progression}
M.~Martcheva and H.~R. Thieme.
\newblock Progression age enhanced backward bifurcation in an epidemic model
  with super-infection.
\newblock {\em Journal of Mathematical Biology}, 46(5):385--424, 2003.

\bibitem{muturi2011larval}
E.~Muturi, C.~Kim, B.~Alto, M.~Berenbaum, and M.~Schuler.
\newblock Larval environmental stress alters adult mosquito fitness and
  competence for arboviruses.
\newblock {\em Trop. Med. Int. Health}, 16:955--964, 2011.

\bibitem{nah2014malaria}
K.~Nah, Y.~Nakata, and G.~R{\"o}st.
\newblock Malaria dynamics with long incubation period in hosts.
\newblock {\em Computers \& Mathematics with Applications}, 68(9):915--930,
  2014.

\bibitem{prosper2014optimal}
O.~Prosper, N.~Ruktanonchai, and M.~Martcheva.
\newblock Optimal vaccination and bednet maintenance for the control of malaria
  in a region with naturally acquired immunity.
\newblock {\em Journal of theoretical biology}, 353:142--156, 2014.

\bibitem{reiner2013systematic}
R.~C. Reiner~Jr, T.~A. Perkins, C.~M. Barker, T.~Niu, L.~F. Chaves, A.~M.
  Ellis, D.~B. George, A.~Le~Menach, J.~R. Pulliam, D.~Bisanzio, et~al.
\newblock A systematic review of mathematical models of mosquito-borne pathogen
  transmission: 1970--2010.
\newblock {\em Journal of The Royal Society Interface}, 10(81):20120921, 2013.

\bibitem{reisen2014effects}
W.~K. Reisen, Y.~Fang, and V.~M. Martinez.
\newblock Effects of temperature on the transmission of west nile virus by
  culex tarsalis (diptera: Culicidae).
\newblock {\em Journal of medical entomology}, 43(2):309--317, 2014.

\bibitem{rock2015age}
K.~Rock, D.~Wood, and M.~Keeling.
\newblock Age-and bite-structured models for vector-borne diseases.
\newblock {\em Epidemics}, 12:20--29, 2015.

\bibitem{salas2015viral}
J.~S. Salas-Benito, D.~Nova-Ocampo, et~al.
\newblock Viral interference and persistence in mosquito-borne flaviviruses.
\newblock {\em Journal of immunology research}, 2015, 2015.

\bibitem{sim2014mosquito}
S.~Sim, N.~Jupatanakul, and G.~Dimopoulos.
\newblock Mosquito immunity against arboviruses.
\newblock {\em Viruses}, 6(11):4479--4504, 2014.

\bibitem{soverow2009infectious}
J.~E. Soverow, G.~A. Wellenius, D.~N. Fisman, and M.~A. Mittleman.
\newblock Infectious disease in a warming world: how weather influenced west
  nile virus in the united states (2001--2005).
\newblock {\em Environmental health perspectives}, 117(7):1049--1052, 2009.

\bibitem{tabachnick2013nature}
W.~Tabachnick.
\newblock Nature, nurture and evolution of intra-species variation in mosquito
  arbovirus transmission competence.
\newblock {\em International journal of environmental research and public
  health}, 10(1):249--277, 2013.

\bibitem{taghikhani2018mathematics}
R.~Taghikhani and A.~B. Gumel.
\newblock Mathematics of dengue transmission dynamics: Roles of vector vertical
  transmission and temperature fluctuations.
\newblock {\em Infectious Disease Modelling}, 3:266--292, 2018.

\bibitem{thieme}
H.~R. Thieme et~al.
\newblock Semiflows generated by lipschitz perturbations of non-densely defined
  operators.
\newblock {\em Differential and Integral Equations}, 3(6):1035--1066, 1990.

\bibitem{tjaden2013extrinsic}
N.~B. Tjaden, S.~M. Thomas, D.~Fischer, and C.~Beierkuhnlein.
\newblock Extrinsic incubation period of dengue: knowledge, backlog, and
  applications of temperature dependence.
\newblock {\em PLoS neglected tropical diseases}, 7(6):e2207, 2013.

\bibitem{tuncer2016structural}
N.~Tuncer, H.~Gulbudak, V.~L. Cannataro, and M.~Martcheva.
\newblock Structural and practical identifiability issues of
  immuno-epidemiological vector--host models with application to rift valley
  fever.
\newblock {\em Bulletin of mathematical biology}, 78(9):1796--1827, 2016.

\bibitem{wang2017global}
X.~Wang, Y.~Chen, and S.~Liu.
\newblock Global dynamics of a vector-borne disease model with infection ages
  and general incidence rates.
\newblock {\em Computational and Applied Mathematics}, pages 1--26, 2017.

\bibitem{webb}
G.~F. Webb.
\newblock {\em Theory of nonlinear age-dependent population dynamics}.
\newblock CRC Press, 1985.

\bibitem{wei2008epidemic}
H.-M. Wei, X.-Z. Li, and M.~Martcheva.
\newblock An epidemic model of a vector-borne disease with direct transmission
  and time delay.
\newblock {\em Journal of Mathematical Analysis and Applications},
  342(2):895--908, 2008.

\bibitem{xu2016hopf}
J.~Xu and Y.~Zhou.
\newblock Hopf bifurcation and its stability for a vector-borne disease model
  with delay and reinfection.
\newblock {\em Applied Mathematical Modelling}, 40(3):1685--1702, 2016.

\bibitem{yadav2005effect}
P.~Yadav, P.~Barde, M.~Gokhale, V.~Vipat, A.~Mishra, J.~Pal, and D.~Mourya.
\newblock Effect of temperature and insecticide stresses on aedes aegypti
  larvae and their influence on the susceptibility of mosquitoes to dengue-2
  virus.

\end{thebibliography}


\begin{table}
\caption{Parameter estimates of within-vector model \eqref{within_vector} fitted to viremia levels in \cite{fortuna2015experimental})} 
\begin{tabularx}{\textwidth}{>{} lXXl}
\toprule
Parameter & Estimate &Units & Reference \\ [0.5ex]
\toprule
\\
$r_v$ & $0.3258$ &$ (\mbox{TCID}_{50}\times \mbox{ days})^{-1}$ & See Section \ref{acc_scales}\\ [0.5 ex] 
$K_v$ & $1.2303\times 10^3$ & TCID$_{50}$ & See Section \ref{acc_scales} \\[0.5 ex]
$U_v$ & $0.9933$& TCID$_{50}$ & See Section \ref{acc_scales}\\ [0.5 ex]
\bottomrule
\end{tabularx}
\label{table:fitparam} 
\end{table}

\begin{table}
\caption{Parameter estimates of within-host model  \eqref{PB} fitted to viremia levels in \cite{tuncer2016structural})} 
\begin{tabularx}{\textwidth}{>{} lXX}
\toprule
Parameter & Estimate &Units\\ [0.5ex]
\toprule
\\
$r$ & $7.21859433$ &$ (\mbox{TCID}_{50}\times \mbox{ days})^{-1}$ \\ [0.5 ex] 
$K$ & $5.828521156794433\times 10^7$ & TCID$_{50}$ \\[0.5 ex]
$a$ & $1.1\times 10^{-7}$& $(\mbox{ELISA PP} \times \mbox{days})^{-1}$\\ [0.5 ex]
$q$ & $ 0.48442451$ &$\mbox{days}^{-1}$ \\ [0.5 ex] 
$w$ &  $0.40599756$ &$\mbox{days}^{-1}$ \\ [0.5 ex]
$b$ & $5\times 10^{-8}$ & $(\mbox{ELISA PP} \times \mbox{days})^{-1}$\\ [0.5 ex]
\bottomrule
\end{tabularx}
\label{table:fitparamPB} 
\end{table}
\begin{table}
\caption{Estimated parameter values of epidemiological model  \eqref{SIRhost_n} - \eqref{SIvector_n} fitted in \cite{tuncer2016structural})} 
\footnotesize{
\begin{tabularx}{\textwidth}{>{} llXllX}
\toprule
Parameter & Fixed Value &Unit&Parameter &Fixed Value&Unit\\ [0.5ex]
\toprule
\\
$\eta$ & ${1}/{40}$ & $\D \frac{\mbox{vector}}{\mbox{time}}$& $I_v(0)$ & $0.000005$ &vector  \\ [2.5 ex]
$\mu$ & ${1}/{40}$ &$\D \frac{1}{\mbox{time}}$&$S_H(0)$ & $0.9999$&host \\[2.5 ex]
$\Lambda$ & ${1}/{(365\times10)}$&$\D \frac{\mbox{host}}{\mbox{time}}$& $i_H(0,\tau)$ & $0.0000001$&$\D \frac{\mbox{host}}{\mbox{time}}$\\ [2.5 ex]
$d$ & $ {1}/{(365\times10)}$ &$\D \frac{1}{\mbox{time}}$& $R_H(0)$ & $0$&host  \\ [2.5 ex] 
$\beta_v$ &  $0.2$ &$\D \frac{1}{\mbox{host}\times\mbox{time}}$&$S_v(0)$ & $0.999995$&vector \\ [2.5 ex]
\bottomrule
\end{tabularx}
}
\label{table:fixparam2} 
\end{table}

\begin{table}
\caption{Parameter estimates of the epidemiological model fitted to human incidences as reported by CDC} 
\footnotesize{
\begin{tabularx}{\textwidth}{>{} lXXlXX}
\toprule
Parameter & Estimate &Unit&Parameter & Estimate&Unit \\ [0.5ex]
\toprule
\\
$b_0$ & $1.609\times 10^{-8}$ &$\D\frac{1}{\mbox{pathogen}\times\mbox{time}}$&$c_0$ &  $0.00296802039$&$\D\frac{\mbox{pathogen}}{\mbox{antibody}\times\mbox{time}}$  \\ [2.5 ex] 
$b_1$ & $8.35602\times 10^{-6}$ & $\D\frac{1}{\mbox{antibody}\times\mbox{time}}$& $\epsilon_0$ & $  2.89724 \times 10^{-6}$&pathogen   \\[2.5 ex]
$a_0$ & $0.68487870728$&$\D\frac{1}{\mbox{host}\times\mbox{time}}$& \\ [2.5 ex]
$a_1$ & $1.7261 \times 10^4$  &pathogen&   \\ [2.5 ex] 
\bottomrule
\end{tabularx}
}
\label{table:fitparamepi} 
\end{table}

\newpage

\end{document}